\documentclass[a4paper,10pt]{article}
\usepackage[top=2.54cm,bottom=2.0cm,left=2.0cm,right=2.54cm, includeheadfoot]{geometry}

\usepackage[T1]{fontenc}
\usepackage[utf8]{inputenc}
\usepackage[english]{babel}
\usepackage{enumerate}
\usepackage{multirow,booktabs}
\usepackage[table]{xcolor}
\usepackage{lastpage}
\usepackage{indentfirst}
\usepackage{subcaption}
\usepackage{float}

\usepackage{footnote}

\usepackage{amsmath,amsfonts,amssymb,amscd,amsthm}

\usepackage{mathtools}

\usepackage{bm}

\usepackage[all,2cell]{xy} \UseAllTwocells \SilentMatrices

\usepackage{mathpartir}

\usepackage{hyperref}

\usepackage{subfiles}

\usepackage{multicol}

\usepackage{lineno}

\usepackage{tikz}

\newtheorem{theorem}{Theorem}[section]
 
\newtheorem{corollary}[theorem]{Corollary}
\newtheorem{proposition}[theorem]{Proposition} 

\newtheorem{definition}[theorem]{Definition} 
\newtheorem{example}[theorem]{Example}

\newtheorem*{claim*}{Claim}

\newtheorem{remark}[theorem]{Remark}

\begin{document}

\title{Decoupling Evidence Sources in Paraconsistent Logic: A Generalized PAL2v Framework for Educational Assessment}

\author{
 Arthur Nakamura $^\textup{\scriptsize c}$
   \and
   David Drummond $^\textup{\scriptsize c}$ \and Welbert Pereira $^\textup{\scriptsize c}$ \and Gustavo Lahr$^\textup{\scriptsize c}$ \and Kaique Roberto$^\textup{\scriptsize a,b}$
}

\date{
   $^\textup{\scriptsize a}$\textit{\small Center for Logic, Epistemology and the History of Science (CLE), State University of Campinas (UNICAMP), Campinas, Brazil}
   \\
   $^\textup{\scriptsize b}$\textit{\small Institute of Philosophy and the Humanities (IFCH),  State University of Campinas (UNICAMP), Campinas, Brazil}
   \\
   $^\textup{\scriptsize c}${\small Faculdade Israelita de Ciências da Saúde Albert Einstein (FICSAE), S\~ao Paulo, Brazil}
}

\maketitle

\begin{abstract}
Paraconsistent Annotated Logic with two values (PAL2v) annotates a proposition with a pair $(\mu,\lambda)$ of degrees of favorable and unfavorable evidence, drawn from a bilattice whose four extreme points are the values of Belnap and Dunn. In applications, however, the second coordinate is almost always taken to be the complement of the first, and several such annotations are then fused. We show that this practice determines the sign of the degree of contradiction in advance, independently of the data: for aggregation functions $F$ and $G$ the fused degree of contradiction equals $F-G^{d}$, where $G^{d}$ is the De Morgan dual of $G$, so that it vanishes identically when the two are dual, and has a constant sign otherwise. In particular the componentwise minimum used throughout the literature yields only consistent annotations, and the conjunction of annotated logic only exact ones; five of the twelve logical states, contradiction among them, are essentially unreachable by construction. We also give a decision procedure for the twelve states that is provably total, in place of the inequality lists in current use, which assign two states to about half of the non-extreme region. We then propose a construction in which the two coordinates are computed from disjoint families of evidence, show that every logical state is attainable. Educational assessment provides a natural instance, since summative and formative instruments are institutionally distinct: mapping examinations to $\mu$ and disengagement to $\lambda$ makes the degree of contradiction equal to the gap between what a student demonstrates and what that student invests. Simulations over three cohort profiles illustrate the classification.

\end{abstract}

\section{Introduction}


Paraconsistent Annotated Logic with two values (PAL2v) is a non-classical logic that allows reasoning with contradictory information while maintaining logical consistency. Formally, let $P$ be a proposition. In PAL2v, $P$ is associated with an annotation $(\mu, \lambda)$ belonging to a lattice $L$, here represented by the unit square $[0, 1]^2$ (or its image contained in $[-1, 1]^2$ given by an adequate bijective linear transformation), where $\mu$ denotes the degree of favorable evidence and $\lambda$ denotes the degree of unfavorable evidence \cite{abe2019three}.

The inferential process in PAL2v applications is governed by the mapping of these annotations onto a barycentric coordinate system, which defines the logical state of $P$. The most widely used operators in the literature are the Degree of Certainty ($D_c$) and the Degree of Contradiction ($D_{ct}$), defined as:
\begin{align*}
    D_c = \mu - \lambda\mbox{ and }D_{ct} = \mu + \lambda - 1
\end{align*}

 The Degree of Certainty $D_c \in [-1, 1]$ measures the proximity of the proposition to the states of absolute truth ($D_c = 1$) or absolute falsity ($D_c = -1$). Conversely, $D_{ct} \in [-1, 1]$ quantifies the degree of inconsistency: when $D_{ct} > 0$, the system identifies a state of contradiction (over-determined information), and when $D_{ct} < 0$, it identifies a state of paracompleteness or indeterminacy (under-determined information) \cite{silvafilho2010uncertainty}. By establishing $D_c$ and $D_{ct}$ as the primary metrics for logical evaluation, the PAL2v framework provides a quantifiable method to navigate the boundary between consistent and inconsistent data. 

Most PAL2v applications are built upon this standard calculation of $D_{ct}$ to identify states of inconsistency or paracompleteness. In the field of Artificial Intelligence, the development of Paraconsistent Artificial Neural Networks (PANnets) uses these operators within their activation functions to process conflicting data without system trivialization \cite{silvafilho2010uncertainty}. This approach has been successfully implemented in medical diagnostics, particularly in the early detection of Alzheimer's disease, where $D_{ct}$ is crucial for reconciling divergent clinical symptoms and laboratory results \cite{silvalopes2010alzheimer}.

Applications typically follow a three-stage pipeline: data normalization, distance-based attribution, and multi-source fusion. Initially, raw clinical or physical indicators $I$ are scaled to the unit interval $[0, 1]$ using min-max normalization:

\begin{align*}
    I_{norm} = \frac{I - I_{min}}{I_{max} - I_{min}}
\end{align*}

Subsequently, the degrees of belief ($\mu$) and disbelief ($\lambda$) are calculated as a function of the distance between the normalized indicator $I_{norm}$ and a reference parameter $P_j$ (such as a statistical mean or a control limit), as observed in the methodologies of Mario et al. \cite{mario2010cephalometric} and Oshiyama et al. \cite{oshiyama2012medical}. The standard attribution follows:

\begin{align*}
    \lambda = |I_{norm} - P_j| \quad \text{and} \quad \mu = 1 - \lambda =1 - |I_{norm} - P_j|
\end{align*}

Finally, when several sensors or expert perspectives are involved, the individual annotations are fused into one. The usual choice is the componentwise minimum,
$(\mu, \lambda)_{final} = (\min\{\mu_1, \dots, \mu_n\}, \min\{\lambda_1, \dots, \lambda_n\})$,
which is commonly described as the conjunction of the annotations, although the conjunction of annotated logic takes the minimum of the degrees of belief and the maximum of the degrees of disbelief \cite{abe2019three}, following the generalized procedure described theoretically as twist structures \cite{busa:gala:mar:2022}. Section~\ref{sec:classical} shows that the difference is not incidental.

This pipeline carries a classical bias. Taking the degree of disbelief to be the complement of the degree of belief \cite{mario2010cephalometric, oshiyama2012medical} makes the two coordinates one quantity in disguise, and the degree of contradiction then records nothing about the evidence: its sign is fixed in advance by the choice of aggregation. It is at most zero for the componentwise minimum, at least zero for the componentwise maximum, and exactly zero for the conjunction of annotated logic. Under any of these choices the fused annotation is confined to one half of the lattice, and five of the twelve logical states are essentially unreachable. A framework adopted in order to represent contradiction thus rules it out at the first step. We make this precise in Section~\ref{sec:classical}.

To address these shortcomings, we propose a generalized PAL2v methodology, focusing on two primary objectives. First, we advocate for representing the logical output directly within the unit square $[0, 1]^2$, as opposed to the traditional transformation into the $[-1, 1]^2$ Cartesian space of certainty and contradiction. This shift is intended to facilitate the visualization and communication of results for non-mathematicians and non-logicians, allowing for a more intuitive grasp of the evidence levels. Second, we propose a novel data-driven approach to calculate $\lambda$ through a process that is structurally independent of $\mu$. By decoupling these parameters, the proposed approach aims to achieve a novel paraconsistent interpretation of the data, where the degrees of contradiction and paracompleteness emerge as properties of the evidence rather than artefacts of the
arithmetic relating the two coordinates.

\section{Related Work}\label{sec:related}
 
The construction proposed here sits at the meeting point of four lines of work:
 
\paragraph{Four-valued and annotated logics.} The idea that a proposition may be supported by evidence for and evidence against, independently, goes back to the four-valued logic of Belnap and Dunn \cite{belnap1977useful}, whose values true, false, both and neither, are exactly the four extreme states of Proposition~\ref{prop:bilattice}. Annotated logics replace the two-element scale by a lattice of annotations, and the two-valued annotated logic PAL2v takes that lattice to be the unit square. What the applied PAL2v literature adds to this picture is a refinement of the four values into twelve, obtained by asking how far an annotation lies from the center of the square and in which direction it leans \cite{abe2019three, silvafilho2010uncertainty}. Our contribution to this line is not the refinement itself, which we take over, but its specification: we replace the lists of inequalities in current use, which do not determine a classification, by a decision procedure that is provably total.
 
\paragraph{Bilattices and twist structures.} The square with its two orders is a bilattice in the sense of Ginsberg \cite{ginsberg1988multivalued}, and the vocabulary developed by Fitting \cite{fitting1991bilattices} --- conflation, exact and consistent elements --- names precisely the loci that our main results turn on. Seen from there, the limitation established in Section~\ref{sec:classical} has a compact statement: the pipeline in current use never leaves the consistent part of the bilattice, and its canonical variant never leaves the exact part, which is a copy of the underlying lattice $[0,1]$ and therefore carries no contradiction at all. That a bilattice of this shape arises as a twist structure over its underlying lattice is well known \cite{busa:gala:mar:2022}; the decoupling of Section~\ref{sec:generalized} can be read as the observation that an application inherits the full twist structure only when it feeds the two coordinates from genuinely separate sources.
 
\paragraph{Applications of PAL2v.} PAL2v has been applied to cephalometric diagnosis \cite{mario2010cephalometric}, to the detection of Alzheimer's disease through paraconsistent neural networks \cite{silvalopes2010alzheimer}, to the management of communication network routes \cite{abe2021sensors} and, closest to the present work, to the classification of medical equipment from corrective maintenance data \cite{oshiyama2012medical}. These applications share a common shape: a single family of indicators is normalized, the distance of each indicator to a reference parameter yields a degree of belief, the degree of disbelief is taken as its complement, and several such annotations are fused. The present article is a study of what that shape costs, and of what a domain must supply in order to avoid it. Educational assessment supplies it, because summative and formative instruments are institutionally distinct, and it is in this sense that the application is not merely one more dataset for an existing method.

\paragraph{Non-classical formalisms in educational assessment.}
Educational assessment has for some time been a testing ground for graded and many-valued formalisms. Fuzzy inference systems have been applied to student performance evaluation widely enough to have attracted a systematic review \cite{tuan2025systematic}, including in engineering education \cite{baba2012fuzzy}; intuitionistic and neutrosophic extensions have followed, adding a degree of hesitation or of indeterminacy to the single membership degree of the fuzzy approach \cite{atanassov1986intuitionistic, smarandache1999unifying}. What these formalisms have in common is that they treat uncertainty as vagueness or as missing information. Contradiction, in the sense of evidence that simultaneously supports and undermines the same claim, is not among the states they are designed to express, and in the intuitionistic case it is excluded by the very constraint $\mu+\lambda\le1$ that defines the framework, as Remark~\ref{rem:atanassov} makes precise. Paraconsistent annotation is the natural formalism for that gap, and it has seen little use in this domain. The present article is concerned with the formal obstacle that has to be removed before it can be applied at all.

\section{Educational context}

In higher education, learning is conceived as a dynamic and multidimensional process that transcends the mere acquisition of technical knowledge. It involves the integration of conceptual understanding, practical skills, and socio-emotional competencies, enabling students to cultivate intellectual autonomy and critical thinking in the face of professional challenges \cite{cetron2019decoding}. One of the central challenges in educational science is to develop methods capable of assessing the evolution of learning in a more comprehensive and nuanced manner. Evaluating only the technical knowledge acquired is insufficient to capture the maturation of perceptions, attitudes, and professional competencies that unfold throughout undergraduate studies \cite{hernandez2025direct}.

In our Biomedical Engineering undergraduate program at the Faculdade Israelita de Ciências da Saúde Albert Einstein, student progress is monitored through summative assessments — such as individual and group tests, essays, seminars, and exams — as well as formative assessments, which include continuous observation of students’ attitudes such as punctuality, teamwork, receptiveness to feedback, and communication skills, complemented by recurrent faculty mentoring. This dual approach seeks to encompass both cognitive development and personal growth, acknowledging that learning is a complex phenomenon integrating technical expertise with human dimensions.
Despite the relevance of combining summative and formative assessments, higher education still lacks robust tools capable of synthesizing these dimensions into actionable insights for educational management and decision-making. The absence of such integrative mechanisms limits the ability to monitor learning trajectories in a holistic manner. Therefore, we consider that a mathematical model designed to predict students’ learning status by jointly analyzing summative and formative indicators could  provide a more comprehensive framework for evaluating learning achievements by academic performance, socio-emotional competencies, and professional attitudes, thereby contributing to more effective educational strategies and institutional decision-making.

\section{The Classical Pipeline and its Structural Limitation}\label{sec:classical}

In the pipeline that has become standard in applications, the degree of unfavorable evidence is obtained as the complement of the favorable one, $\lambda_k=1-\mu_k$, and several such annotations are then fused by a pair of aggregation functions. The choice of that pair is usually presented as a modeling decision. The following observation shows that it is not: once each source satisfies $\lambda_k=1-\mu_k$, the pair of aggregators fixes the sign of the degree of contradiction in advance, independently of the data.

\begin{definition}[De Morgan dual]\label{def:dual}
Let $G:[0,1]^n\to[0,1]$. The \emph{De Morgan dual} of $G$ is the function
$G^{d}:[0,1]^n\to[0,1]$ given by
$$G^{d}(z_1,\dots,z_n)=1-G(1-z_1,\dots,1-z_n).$$
\end{definition}

\begin{proposition}[The sign of the contradiction is algebraic]\label{prop:duality}
Let $F,G:[0,1]^n\to[0,1]$ be aggregation functions, and let
$(\mu_k,\lambda_k)$, $k=1,\dots,n$, be annotations with $\lambda_k=1-\mu_k$ for every $k$. Write $\mu^{\ast}=F(\mu_1,\dots,\mu_n)$ and $\lambda^{\ast}=G(\lambda_1,\dots,\lambda_n)$ for the fused annotation. Then, for every $\mu=(\mu_1,\dots,\mu_n)$,
$$D_{ct}(\mu^{\ast},\lambda^{\ast})=F(\mu)-G^{d}(\mu).$$
Consequently:
\begin{enumerate}
  \item $D_{ct}$ vanishes identically if and only if $G$ is the De Morgan dual of
        $F$;
  \item if $F\le G^{d}$ pointwise, then $D_{ct}\le0$ for every input;
  \item if $F\ge G^{d}$ pointwise, then $D_{ct}\ge0$ for every input.
\end{enumerate}
\end{proposition}

\begin{proof}
Since $\lambda_k=1-\mu_k$ we have
$\lambda^{\ast}=G(1-\mu_1,\dots,1-\mu_n)=1-G^{d}(\mu)$ by
Definition~\ref{def:dual}. Therefore
$$D_{ct}(\mu^{\ast},\lambda^{\ast})=\mu^{\ast}+\lambda^{\ast}-1
   =F(\mu)+\bigl(1-G^{d}(\mu)\bigr)-1=F(\mu)-G^{d}(\mu),$$
which is the stated identity. Items 1 to 3 are immediate readings of it.
\end{proof}

Proposition~\ref{prop:duality} applies in particular to the three aggregation pairs that occur in the literature, and none of them escapes.

\begin{corollary}\label{cor:three-pairs}
Under the hypotheses of Proposition~\ref{prop:duality}:
\begin{enumerate}
  \item if $F=\min$ and $G=\max$ --- the conjunction of annotated logic --- then
        $D_{ct}=0$ identically;
  \item if $F=G=\min$, then
        $D_{ct}=\min_{k}\mu_k-\max_{k}\mu_k\le0$, with equality if and only if
        $\mu_1=\dots=\mu_n$;
  \item if $F=G=\max$, then
        $D_{ct}=\max_{k}\mu_k-\min_{k}\mu_k\ge0$, with equality under the same
        condition.
\end{enumerate}
\end{corollary}

\begin{proof}
For $G=\max$ we have $G^{d}(\mu)=1-\max_k(1-\mu_k)=\min_k\mu_k$, so with $F=\min$ Proposition~\ref{prop:duality} gives $D_{ct}=\min_k\mu_k-\min_k\mu_k=0$. For $G=\min$ we have $G^{d}(\mu)=\max_k\mu_k$, so $F=\min$ gives $D_{ct}=\min_k\mu_k-\max_k\mu_k$, which is non-positive and vanishes exactly when the minimum equals the maximum; and $F=\max$ gives $D_{ct}=\max_k\mu_k-\min_k\mu_k$, which is the same quantity with the opposite sign.
\end{proof}

The three cases of Corollary~\ref{cor:three-pairs} should be compared directly. The canonical conjunction places the fused annotation exactly on the anti-diagonal $\mu+\lambda=1$, where the degree of contradiction is zero by construction. The pairwise minimum places it on or below that line, and the pairwise maximum on or above it. In all three cases the quantity that is being read as a degree of contradiction is, up to sign, the \emph{spread} $\max_k\mu_k-\min_k\mu_k$ of the sources: it measures how much the sources disagree with one another, not how much the evidence about $P$ conflicts. This matters, because disagreement among sources and conflict in the evidence are different phenomena, and a methodology that reports the first under the name of the second cannot detect the second at all.

\begin{corollary}\label{cor:unreachable}
Under the hypotheses of Proposition~\ref{prop:duality} with $F=G=\min$, the fused annotation always lies on or below the anti-diagonal $\mu+\lambda=1$ of the USCP. Hence
\begin{enumerate}
  \item the states $\top$, $\top\!\to\!t$ and $\top\!\to\!f$, which require $\mu+\lambda>1$, are unreachable;
  \item the states $qt\!\to\!\top$ and $qf\!\to\!\top$, which require $\mu+\lambda\ge1$, are reachable only when $\mu_1=\dots=\mu_n$, in which case the fused annotation lies on the anti-diagonal.
\end{enumerate}
\end{corollary}
\begin{proof}
By Corollary~\ref{cor:three-pairs}, $\mu^{\ast}+\lambda^{\ast}-1\le0$. For the first three states, the condition $\mu+\lambda>1$ would require $\mu^{\ast}+\lambda^{\ast}>1$, impossible. For the last two, the condition $\mu+\lambda\ge1$ can hold only if $\mu^{\ast}+\lambda^{\ast}=1$, which by Corollary~\ref{cor:three-pairs} occurs exactly when $\mu_1=\dots=\mu_n$.
\end{proof}

\begin{remark}[Half of the square is unused]\label{rem:atanassov}
Corollary~\ref{cor:unreachable} says that this pipeline confines every annotation to the consistent part of the bilattice, so that half of the diagram --- and five of the twelve logical states --- is unreachable in the non-degenerate case, by construction. That half has an independent life in the literature: the set $\{\mu+\lambda\le1\}$ is exactly the domain of Atanassov's intuitionistic fuzzy sets, in which a membership degree and a non-membership degree are constrained to sum to at most one, and the slack
$$\pi=1-\mu-\lambda$$
is the degree of hesitation, recording how much evidence is missing \cite{atanassov1986intuitionistic}. Two remarks follow. First, a construction introduced in order to represent contradiction ends up operating, without saying so, inside a framework designed to exclude it. Second, and less obviously, the hesitation degree is not free to do its work either: by Corollary~\ref{cor:three-pairs}, under the componentwise minimum $\pi=\max_k\mu_k-\min_k\mu_k$, so $\pi$ measures the spread among the sources rather than the evidence they fail to supply. Both readings of the lower half of the square are thereby lost at once. Under the construction of Section~\ref{sec:generalized} the annotation ranges over the whole square, $\pi$ recovers its intended meaning on the consistent half, and its negative values on the other half record conflict rather than absence.
\end{remark}

\begin{remark}[On the specification of the twelve regions]\label{rem:overlap}
A second difficulty is independent of the aggregation step and concerns the specification of the states themselves. The twelve regions are usually specified by lists of inequalities on $D_c$, $D_{ct}$, $\mu$ and $\lambda$, rather than by a decision procedure. Such lists do not determine a classification: evaluated on a uniform grid over the USCP, they assign two distinct states to approximately half of the non-extreme region, and they leave the dividing lines themselves unassigned, since the conditions on either side are strict. The annotation $(\mu,\lambda)=(0.6,\,0.3)$ is a case in point: it satisfies simultaneously the usual conditions for $qt\!\to\!\top$ and for $qt\!\to\!\bot$. Under
Definition~\ref{def:decision} it lies in the quadrant $u>0$, $v<0$, at distance $|u|+|v|=0.3<\tfrac12$ from the center, with $\mu+\lambda=0.9<1$, and its state is unambiguously $qt\!\to\!\bot$. Proposition~\ref{prop:partition} guarantees that no such ambiguity remains.
\end{remark}

\section{The paraconsistent annotation and its logical states}

Throughout this work, the proposition under evaluation is

\begin{quote}
$P_i$: \emph{``student $i$ has attained the intended learning outcomes of the course.''}
\end{quote}

Summative assessments provide evidence \emph{in favor} of $P_i$: a student who performs well on examinations gives us reason to assert $P_i$. Formative assessments provide evidence \emph{against} $P_i$ in the following precise sense: the intended learning outcomes of the curriculum include preparation, teamwork and participation, so a student who systematically fails to engage gives us reason to deny $P_i$, \emph{independently} of examination performance. The two sources are therefore evidence about the same proposition, gathered through different instruments, and neither determines the other. This is exactly the situation that paraconsistent annotation is designed to represent, and it is why the annotation cannot be reduced to a single aggregate grade.

\begin{definition}[Annotation]\label{def:annotation}
An \emph{annotation} is a pair $(\mu,\lambda)\in[0,1]^2$, where $\mu$ is the degree of favorable evidence and $\lambda$ the degree of unfavorable evidence for $P_i$. The set $[0,1]^2$ is called the Unit Square of the Cartesian Plane (USCP). We write
$$u=\mu-\tfrac12, \qquad v=\lambda-\tfrac12$$
for the position of the annotation relative to the center of the USCP.
\end{definition}

All annotations, all figures and all results in this work are expressed in the USCP. The degrees of certainty $D_c=\mu-\lambda$ and of contradiction $D_{ct}=\mu+\lambda-1$ are used only as abbreviations inside proofs; all operations remain within the USCP.
 
The square carries two main orderings, and it is worth making them explicit before the classification is defined, since the states of Definition~\ref{def:decision} are named after their extreme points.
 
\begin{definition}[Truth and knowledge orders]\label{def:orders}
For annotations $a=(\mu_1,\lambda_1)$ and $b=(\mu_2,\lambda_2)$ we set
$$a\le_t b\mbox{ if, and only if, }\mu_1\le\mu_2\mbox{ and }\lambda_1\ge\lambda_2;$$
and
$$a\le_k b\mbox{ if, and only if, }\mu_1\le\mu_2 \mbox{ and } \lambda_1\le\lambda_2 .$$
We read $a\le_t b$ as ``$b$ is at least as true as $a$'' and $a\le_k b$ as ``$b$ rests on at least as much evidence as $a$''.
\end{definition}
 
\begin{proposition}\label{prop:bilattice}
$(\,[0,1]^2,\le_t,\le_k\,)$ is a complete interlaced bilattice, namely the product bilattice of the unit interval with itself. Its operations are
$$a\wedge_t b=(\mu_1\wedge\mu_2,\;\lambda_1\vee\lambda_2),
\qquad a\vee_t b=(\mu_1\vee\mu_2,\;\lambda_1\wedge\lambda_2),$$
$$a\otimes b=(\mu_1\wedge\mu_2,\;\lambda_1\wedge\lambda_2),
\qquad a\oplus b=(\mu_1\vee\mu_2,\;\lambda_1\vee\lambda_2),$$
where $\wedge$ and $\vee$ denote minimum and maximum in $[0,1]$. The extreme points of the two orders are
$$t=(1,0),\quad f=(0,1),\quad \top=(1,1),\quad \bot=(0,0),$$
$t$ and $f$ being the greatest and least elements for $\le_t$, and $\top$ and $\bot$ for $\le_k$. The four of them form a sub-bilattice isomorphic to Belnap-Dunn four-valued lattice (\cite{belnap1977useful}, \cite{fitting1991bilattices}, \cite{ginsberg1988multivalued}).
\end{proposition}
 
\begin{proof}
Both orders are products of the order of $[0,1]$ with itself, the second coordinate being reversed in the case of $\le_t$; each is therefore a complete lattice with the operations displayed. Interlacing amounts to the monotonicity of each pair of operations with respect to the other order, which holds because every operation is computed coordinatewise from $\wedge$ and $\vee$ on $[0,1]$, and these are monotone. The four listed points are the extreme points of the two coordinatewise orders, and restricting to $\{0,1\}^2$ gives the Belnap-Dunn lattice.
\end{proof}
 
Two involutions act on the square, and both have an immediate reading in the present application.
 
\begin{definition}[Negation and Conflation]\label{def:involutions}
For $a=(\mu,\lambda)$ we define
$$\neg a:=(\lambda,\mu), \qquad -a:=(1-\lambda,\,1-\mu).$$
\end{definition}
 
\begin{proposition}\label{prop:involutions}
Negation is an involution that reverses $\le_t$ and preserves $\le_k$; conflation is an involution that preserves $\le_t$ and reverses $\le_k$. Geometrically, $\neg$ is the reflection of the USCP in the diagonal $\mu=\lambda$ and $-$ is the reflection in the anti-diagonal $\mu+\lambda=1$; their composite
$a\mapsto(1-\mu,1-\lambda)$ is the reflection in the center. Consequently $\{\mathrm{id},\neg,-,\neg-\}$ is a Klein four-group of symmetries of the square.
\end{proposition}
 
\begin{proof}
Both maps are involutive. If $a\le_t b$ then $\mu_1\le\mu_2$ and
$\lambda_1\ge\lambda_2$, so $\neg a=(\lambda_1,\mu_1)\ge_t(\lambda_2,\mu_2)=\neg b$, while $a\le_k b$ gives $\neg a\le_k\neg b$ directly. For conflation, $a\le_t b$ gives $1-\lambda_1\le1-\lambda_2$ and $1-\mu_1\ge1-\mu_2$, that is $-a\le_t-b$; and $a\le_k b$ gives $-a\ge_k-b$. The geometric descriptions are immediate, and the composite of two reflections in perpendicular lines is the reflection in their intersection.
\end{proof}
 
The vocabulary of bilattice theory names the two loci that the following sections turn on. An annotation $a$ is called \emph{exact} when $a=-a$ and \emph{consistent} when $a\le_k-a$ (\cite{fitting1991bilattices}). Unwinding the definitions, $a$ is exact precisely when $\mu+\lambda=1$ and consistent precisely when $\mu+\lambda\le1$; the exact annotations form the anti-diagonal, which is a copy of $[0,1]$ under $\le_t$ and on which $D_{ct}$ vanishes, and the consistent ones fill the half of the square below the anti-diagonal.
 
\begin{remark}[The fusion of the classical pipeline is consensus, not conjunction]
\label{rem:consensus}
Proposition~\ref{prop:bilattice} settles a point of terminology that Section~\ref{sec:classical} depends on. The componentwise minimum $(\min_k\mu_k,\min_k\lambda_k)$ used to fuse annotations in applications is $\bigotimes_k a_k$, the meet of the \emph{knowledge} order, which combines sources by retaining only what they agree on. The conjunction of annotated logic is $\bigwedge_{t}$, the meet of the \emph{truth} order, which takes the minimum of the degrees of belief and the \emph{maximum} of the degrees of disbelief. The two coincide only when all sources agree. Describing $\otimes$ as conjunction is therefore a misnomer, and not a harmless one: Corollary~\ref{cor:three-pairs} shows that under $\lambda_k=1-\mu_k$ the operator $\otimes$ returns a consistent annotation for every input, and $\wedge_t$ an exact one.
\end{remark}

The classification of an annotation is settled by three questions, each of which has a direct reading on the square and can be theorized in an abstract setting. In this paper, two of the three are fixed: the quadrants are cut by $\mu=\tfrac12$ and $\lambda=\tfrac12$, and each quadrant is split by the diagonals $\mu=\lambda$ and $\mu+\lambda=1$. The third is not. How far from the center an annotation must lie before an outright verdict is issued is not settled by the logic, and the answer may reasonably depend on the assessment practice of the curriculum. We therefore leave it as a parameter.

\begin{definition}[Extremity index]\label{def:extremity}
An \emph{extremity index} is a continuous function $E:[0,1]^2\to[0,\infty)$ such that $E(\tfrac12,\tfrac12)=0$ and, for every direction $w$ with $\|w\|=1$, the function $s\mapsto E\bigl((\tfrac12,\tfrac12)+sw\bigr)$ is non-decreasing on $\{s\ge0:(\tfrac12,\tfrac12)+sw\in[0,1]^2\}$. Together with a threshold $r>0$, it declares an annotation \emph{extreme} when $E(\mu,\lambda)\ge r$ and \emph{non-extreme} otherwise.
\end{definition}

The condition says that the sublevel sets $\{E<r\}$ are star-shaped neighborhoods of the center of the square. Nothing more is required: no convexity, no triangle inequality and no homogeneity, so $E$ need not be a norm. The \emph{canonical} choice, used in every figure of this article unless stated otherwise, is $E(\mu,\lambda)=|u|+|v|$ with $r=\tfrac12$, whose boundary is the diamond inscribed in the USCP with vertices $(\tfrac12,0)$, $(1,\tfrac12)$, $(\tfrac12,1)$ and $(0,\tfrac12)$.

\begin{definition}[Decision procedure]\label{def:decision}
Let $(\mu,\lambda)\in[0,1]^2$ and let $u,v$ be as in Definition~\ref{def:annotation}.

\noindent\textbf{Step 1 (the four logical states).}
The first applicable clause fixes the \emph{dominant} logical state:
\begin{itemize}
  \item $u\ge0$ and $v\le0$ \quad(high favorable, low unfavorable)\quad $\longrightarrow$\quad $t$;
  \item $u\le0$ and $v\ge0$ \quad(low favorable, high unfavorable)\quad $\longrightarrow$\quad $f$;
  \item $u>0$ and $v>0$ \quad(both high)\quad $\longrightarrow$\quad $\top$;
  \item $u<0$ and $v<0$ \quad(both low)\quad $\longrightarrow$\quad $\bot$.
\end{itemize}

\textbf{Step 2 (how far: extreme or not).} Fix an extremity index $E$ and a threshold $r>0$. The annotation is \emph{extreme} when
$$E(\mu,\lambda)\;\ge\;r ,$$
and in that case the logical state is the dominant state of Step~1: $t$, $f$, $\top$ or $\bot$. Under the canonical choice this reads $|u|+|v|\ge\tfrac12$, that is, the annotation lies on or outside the inscribed polygon.

\textbf{Step 3 (in which direction).}
If $E(\mu,\lambda)<r$ the annotation is \emph{non-extreme} and the dominant state is qualified by the sign of the remaining coordinate:
\begin{itemize}
  \item dominant $t$ or $f$: the state is $qt$ or $qf$ (\emph{quasi-true}, \emph{quasi-false}), tending to $\top$ if $\mu+\lambda\ge1$ and to $\bot$ if $\mu+\lambda<1$;
  \item dominant $\top$ or $\bot$: the state tends to $t$ if $\mu\ge\lambda$ and to $f$ if $\mu<\lambda$.
\end{itemize}
\end{definition}

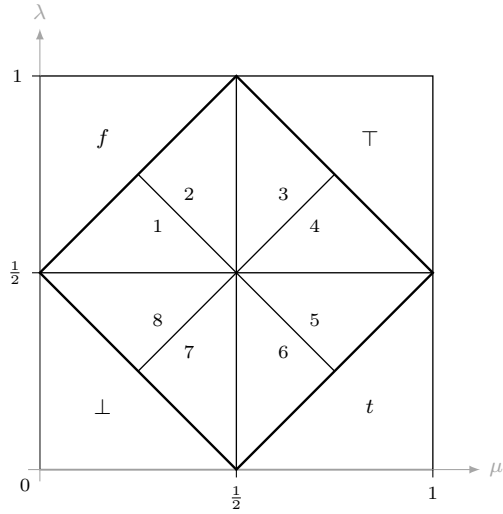
\begin{figure}[htpb]
\centering
\begin{tikzpicture}[scale=5.2,
   axline/.style={-latex,gray!70,line width=0.4pt},
   cut/.style={line width=0.5pt,black},
   dia/.style={line width=0.9pt,black},
   lbl/.style={font=\footnotesize},
   num/.style={font=\scriptsize}]

  \draw[cut] (0,0) rectangle (1,1);

  \draw[axline] (-0.03,0) -- (1.12,0) node[right,lbl]{$\mu$};
  \draw[axline] (0,-0.03) -- (0,1.12) node[above,lbl]{$\lambda$};
  \draw (0.5,0) -- (0.5,-0.02) node[below,num]{$\tfrac12$};
  \draw (1,0) -- (1,-0.02) node[below,num]{$1$};
  \draw (0,0.5) -- (-0.02,0.5) node[left,num]{$\tfrac12$};
  \draw (0,1) -- (-0.02,1) node[left,num]{$1$};
  \node[num,below left] at (0,0) {$0$};

  \draw[cut] (0.5,0) -- (0.5,1);
  \draw[cut] (0,0.5) -- (1,0.5);
  \draw[cut] (0.25,0.25) -- (0.75,0.75);
  \draw[cut] (0.25,0.75) -- (0.75,0.25);

  \draw[dia] (0.5,0) -- (1,0.5) -- (0.5,1) -- (0,0.5) -- cycle;

  \node[lbl] at (0.16,0.16) {$\bot$};
  \node[lbl] at (0.84,0.16) {$t$};
  \node[lbl] at (0.16,0.84) {$f$};
  \node[lbl] at (0.84,0.84) {$\top$};

  \node[num] at (0.30,0.62) {1};
  \node[num] at (0.38,0.70) {2};
  \node[num] at (0.62,0.70) {3};
  \node[num] at (0.70,0.62) {4};
  \node[num] at (0.70,0.38) {5};
  \node[num] at (0.62,0.30) {6};
  \node[num] at (0.38,0.30) {7};
  \node[num] at (0.30,0.38) {8};

\end{tikzpicture}
\caption{The twelve logical states on the USCP. The horizontal axis is the degree of favorable evidence $\mu$ and the vertical axis the degree of unfavorable evidence $\lambda$. The extremity boundary is drawn for the canonical choice $E=|u|+|v|$, $r=\tfrac12$: the bold diamond is the locus $|\mu-\tfrac12|+|\lambda-\tfrac12|=\tfrac12$: annotations on or outside it are extreme and receive one of the four states $t$, $f$, $\top$, $\bot$ shown in the corner triangles; annotations strictly inside it are non-extreme and receive one of the eight states numbered $1$--$8$, identified in Table~\ref{tab:states}. The four segments meeting at the center are $\mu=\tfrac12$, $\lambda=\tfrac12$, $\mu=\lambda$ and $\mu+\lambda=1$}
\label{fig:uscp-regions}
\end{figure}

Figure~\ref{fig:uscp-regions} makes the three steps of Definition~\ref{def:decision} visible at once. The two segments $\mu=\tfrac12$ and $\lambda=\tfrac12$ split the square into the four quadrants of Step~1, and each quadrant carries one logical state: the lower right quadrant, where favorable evidence is high and unfavorable evidence is low, carries truth; the upper left quadrant carries falsity; the upper right quadrant, where both kinds of evidence are high, carries inconsistency; and the lower left quadrant, where both are low, carries paracompleteness. The bold diamond is the threshold of Step~2: an annotation lying on or outside it is far enough from the center for its state to be asserted outright, and the four corner triangles it cuts off are precisely the four extreme states. Inside the diamond the evidence is too weak for an outright verdict, and the two remaining segments $\mu=\lambda$ and $\mu+\lambda=1$ divide each quadrant into the two sectors of Step~3, according to the direction in which the annotation leans.

Reading the figure clockwise from the top left therefore recovers the eight non-extreme states in the order $1,\dots,8$ of Table~\ref{tab:states}. Two features of the diagram deserve emphasis for the discussion that follows. First, the anti-diagonal $\mu+\lambda=1$ separates the annotations that carry an excess of evidence, above it, from those that carry a deficit, below it; the five states at or above the anti-diagonal are exactly those that the classical pipeline cannot reach in the non-degenerate case (Corollary~\ref{cor:unreachable}). Second, no student is ever left unclassified: by Proposition~\ref{prop:partition} the twelve regions cover the square without overlap, so every pair $(\mu,\lambda)$ receives exactly one state and one corresponding pedagogical reading.

\begin{table}[htpb]
\centering
\caption{The twelve logical states, their defining conditions and their reading for a student with annotation $(\mu,\lambda)$. Regions $1$--$8$ are numbered as in Figure~\ref{fig:uscp-regions}; $E$ is the extremity index and $r$ the threshold of Definition~\ref{def:extremity}, whose canonical values are $E=|u|+|v|$ and $r=\tfrac12$; the conditions are stated on the unit square, and Remark~\ref{rem:scaling} gives the corresponding thresholds on the $[0,10]$ axes used in Section~\ref{sec:results}.}
\label{tab:states}
\small
\begin{tabular}{@{}llll@{}}
\toprule
Region & Conditions & State & Reading \\
\midrule
--- & $E\ge r$, $\mu\ge\tfrac12$, $\lambda\le\tfrac12$
    & $t$ & outcomes attained \\
--- & $E\ge r$, $\mu\le\tfrac12$, $\lambda\ge\tfrac12$
    & $f$ & outcomes not attained \\
--- & $E\ge r$, $\mu>\tfrac12$, $\lambda>\tfrac12$
    & $\top$ & strong exams, weak engagement \\
--- & $E\ge r$, $\mu<\tfrac12$, $\lambda<\tfrac12$
    & $\bot$ & weak exams, strong engagement \\
\midrule
1 & $E<r$, $\mu\le\tfrac12$, $\lambda\ge\tfrac12$, $\mu+\lambda<1$
  & $qf\!\to\!\bot$ & leaning negative, evidence lacking \\
2 & $E<r$, $\mu\le\tfrac12$, $\lambda>\tfrac12$, $\mu+\lambda\ge1$
  & $qf\!\to\!\top$ & leaning negative, evidence conflicting \\
3 & $E<r$, $\mu>\tfrac12$, $\lambda>\tfrac12$, $\mu<\lambda$
  & $\top\!\to\!f$ & conflicting, engagement dominating \\
4 & $E<r$, $\mu>\tfrac12$, $\lambda>\tfrac12$, $\mu\ge\lambda$
  & $\top\!\to\!t$ & conflicting, exams dominating \\
5 & $E<r$, $\mu\ge\tfrac12$, $\lambda\le\tfrac12$, $\mu+\lambda\ge1$
  & $qt\!\to\!\top$ & leaning positive, evidence conflicting \\
6 & $E<r$, $\mu\ge\tfrac12$, $\lambda\le\tfrac12$, $\mu+\lambda<1$
  & $qt\!\to\!\bot$ & leaning positive, evidence lacking \\
7 & $E<r$, $\mu<\tfrac12$, $\lambda<\tfrac12$, $\mu\ge\lambda$
  & $\bot\!\to\!t$ & evidence lacking, exams dominating \\
8 & $E<r$, $\mu<\tfrac12$, $\lambda<\tfrac12$, $\mu<\lambda$
  & $\bot\!\to\!f$ & evidence lacking, engagement dominating \\
\bottomrule
\end{tabular}
\end{table}

Read on the square, the procedure says: the quadrant tells which of the four logical states is in play; the distance from the center tells whether the evidence is strong enough for that state to be asserted outright; and, when it is not, the remaining coordinate tells in which direction the annotation leans. In particular, the diamond drawn in every figure of this article is not decoration: it is the boundary between extreme and non-extreme states under the canonical choice of Definition~\ref{def:extremity}.

\begin{example}[A student, step by step]\label{ex:student}
Consider a student whose four summative scores, on the $[0,10]$ scale, are $9.0$, $8.6$, $8.0$ and $8.2$, and whose three formative scores are $3.0$, $2.4$ and $3.6$. Taking the weights to be uniform, for the sake of a self-contained illustration (the weighting of Remark~\ref{rem:weights} is defined at cohort level and does not apply to a single student in isolation), the two averages are $\overline{x}=0.845$ and $\overline{y}=0.300$, so by Definition~\ref{def:generalized} the annotation is
$$(\mu,\lambda)=(0.845,\;0.700).$$

\emph{Step 1.} Both coordinates exceed $\tfrac12$, so the dominant state is $\top$: the evidence for and the evidence against are simultaneously strong.
 
\emph{Step 2.} The canonical extremity index gives
$E=|0.845-\tfrac12|+|0.700-\tfrac12|=0.545\ge\tfrac12$, so the annotation is extreme and the state is $\top$ outright, with no qualifying tendency.
 
\emph{Reading.} By Proposition~\ref{prop:reading}, $D_{ct}=\overline{x}-\overline{y}=0.545$: the gap between examination performance and course engagement is more than half the scale. The student is the case that Table~\ref{tab:states} records as strong examinations with weak engagement, and the appropriate response concerns participation rather than academic support.
 
Running the classical pipeline on the same student is revealing. Treating each of the seven instruments as a source with $\lambda_k=1-\mu_k$ and fusing by the componentwise minimum gives $\mu^{\ast}=0.240$ and $\lambda^{\ast}=0.100$, whence $\mu^{\ast}+\lambda^{\ast}=0.340\le1$, as Corollary~\ref{cor:unreachable} requires. The resulting state is $\bot$: the pipeline reports an absence of evidence about a student on whom seven separate instruments have in fact reported. The contradiction that the generalized construction makes visible is not weakened by the classical pipeline --- it is relocated to the opposite corner of the square.
\end{example}

\begin{proposition}[Partition]\label{prop:partition}
The procedure of Definition~\ref{def:decision} provides a function $\sigma:[0,1]^2\to S$ where
$$S=\{\,t,\;f,\;\top,\;\bot,\;qt\!\to\!\top,\;qt\!\to\!\bot,\;
      qf\!\to\!\top,\;qf\!\to\!\bot,\;\top\!\to\!t,\;\top\!\to\!f,\;
      \bot\!\to\!t,\;\bot\!\to\!f\,\}.$$ 
    Consequently $\{\sigma^{-1}(\varsigma)\}_{\varsigma\in S}$ is a partition of $[0,1]^2$ into at most twelve regions. If, in addition, the extremity index $E$ is positive on the four corners $(0,0),(1,0),(0,1),(1,1)$ and $0<r<\min\{E(0,0),E(1,0),E(0,1),E(1,1)\}$, then the twelve regions are non-empty. Under the canonical choice, each region is moreover the intersection of $[0,1]^2$ with finitely many half-planes, distinct regions have disjoint interiors, and their closures are the twelve closed regions of Figure~\ref{fig:uscp-regions}.
\end{proposition}

\begin{proof}
Step~1 is exclusive and exhaustive. Indeed, the clauses of Step~1 are applied in order, so exclusiveness is automatic; for exhaustiveness, suppose neither of the first two clauses applies. If $u=0$ then $v\le0$ fails, so $v>0$, and then $u\le0$ and $v\ge0$ hold, so the second clause would have applied --- a contradiction. Hence $u\neq0$, and symmetrically $v\neq0$. Since neither of the first two clauses applies, $u$ and $v$ have the same sign, so exactly one of the last two clauses applies. Steps~2 and~3 are dichotomies: the first is $E\ge r$ or $E<r$, the second a comparison of signs. Therefore $\sigma$ is a function, and its fibers partition $[0,1]^2$, whatever $E$ and $r$ may be. 

Since $E$ is continuous and $E(\tfrac12,\tfrac12)=0$, the sublevel set $\{E<r\}$ is an open neighborhood of the center and intersects each of the eight non-extreme sectors. For each corner $c$, continuity of $E$ and $E(c)>r$ give a neighborhood of $c$ on which $E>r$, intersecting the corresponding extreme quadrant. Hence all twelve regions are non-empty.

Under the canonical choice, the fibers are the twelve regions of the figure. Here each test compares two affine functions of $(\mu,\lambda)$, so each fiber is a finite intersection of half-planes. Write $d=D_c=u-v$ and $s=D_{ct}=u+v$. Two identities do all the work:
$$
  \max\{|d|,|s|\}=|u|+|v|,
  \qquad
  |d|\ge|s|\mbox{ if, and only if, } uv\le0 .
$$
The first holds because $\max\{|u-v|,|u+v|\}=|u|+|v|$ for all reals; the second because $d^2-s^2=-4uv$. By the second identity, the two certainty quadrants of Step~1 are precisely the region where certainty dominates contradiction, and the two remaining quadrants are precisely the region where contradiction dominates. By the first identity, Step~2 says exactly that $\max\{|d|,|s|\}\ge\tfrac12$, which cuts off the four corner triangles of the USCP. The four lines $u=0$, $v=0$, $u=v$ and $u=-v$ then divide the remaining diamond into the eight sectors distinguished in Step~3, one for each pairing of a dominant state with a tendency. Every state is attained because each of the twelve regions so obtained has non-empty interior.
\end{proof}

The proof of the first half uses nothing about the shape of $E$, only that Step~2 is a dichotomy. The partition is therefore guaranteed for any choice of boundary, and that choice can be made on pedagogical grounds without any risk to the classification.

\begin{example}[The $\ell^p$ family]\label{ex:lp}
A convenient one-parameter family of extremity indices is
$$E_p(\mu,\lambda)=\bigl(|u|^p+|v|^p\bigr)^{1/p},\qquad p\in(0,\infty],$$
with $E_\infty=\max\{|u|,|v|\}$. The canonical choice is $p=1$. For $p>1$ the boundary bulges outwards and fewer annotations are extreme; for $p<1$ it curves inwards and more are, the index then failing the triangle inequality while remaining an extremity index in the sense of Definition~\ref{def:extremity}. A further natural choice, outside this family, is $E_{\min}(\mu,\lambda)=\min\{|u|,|v|\}$, whose extreme region consists of four squares at the corners of the USCP. Figure~\ref{fig:extremity-family} shows four cases.
\end{example}

\begin{figure}[htpb]
    \centering
    \includegraphics[width=\linewidth]{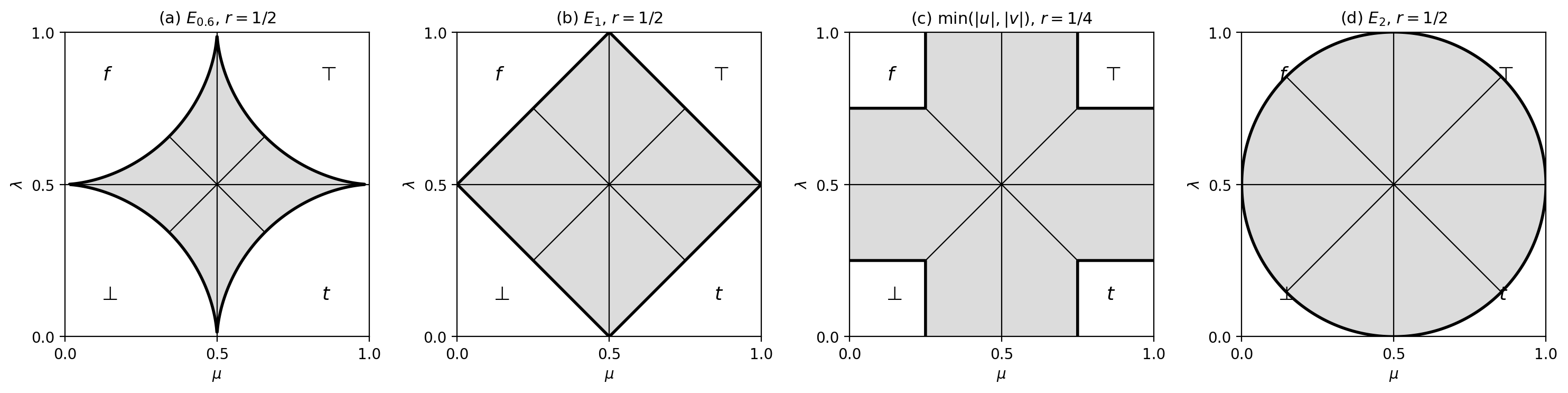}
    \caption{Four extremity boundaries on the USCP, with the non-extreme region shaded. (a) $E_{0.6}$ with $r=\tfrac12$, whose boundary curves inwards; (b) the canonical choice $E_1$ with $r=\tfrac12$; (c) $E_{\min}$ with $r=\tfrac14$, whose extreme region consists of four corner squares; (d) $E_2$ with $r=\tfrac12$. Steps~1 and~3 are unaffected: the quadrants and the two diagonals are the same in all four cases}
    \label{fig:extremity-family}
\end{figure}

\begin{remark}[Calibration]\label{rem:calibration}
The index and the threshold must be chosen together, since they jointly determine how demanding the classification is. With $r=\tfrac12$ and the $\ell^p$ family, the proportion of the USCP in which an outright verdict is issued is about three quarters at $p=0.6$, one half at $p=1$ and about one fifth at $p=2$, and it vanishes at $p=\infty$; with $E_{\min}$ and $r=\tfrac14$ it is one quarter. The choice is a modeling decision, and it should be stated explicitly and justified against the assessment practice of the curriculum.
\end{remark}

\begin{remark}[Corner-shaped boundaries]\label{rem:corner}
Definition~\ref{def:extremity} requires only that $E$ be non-decreasing along rays. The index $E_{\min}$ satisfies this but not the strict version: it is constant along the two median lines $\mu=\tfrac12$ and $\lambda=\tfrac12$. The consequence is behavioral rather than technical. Under $E_{\min}$ an annotation on a median line is never extreme, however far from the center it lies: a student with maximal summative performance and exactly average engagement receives a qualified rather than an outright verdict. Indices satisfying the strict condition, such as the $\ell^p$ family, do not have this feature. Whether it is a defect depends on whether a curriculum wishes outright verdicts to require both dimensions to be decisive, which is what $E_{\min}$ encodes.
\end{remark}

The two involutions of Definition~\ref{def:involutions} descend to the twelve states, provided the extremity boundary respects them.
 
\begin{definition}[Symmetric extremity index]\label{def:symmetric}
An extremity index $E$ is \emph{symmetric} when $E(\neg a)=E(a)$ and $E(-a)=E(a)$ for every annotation $a$.
\end{definition}
 
Every index considered in this article is symmetric: the members $E_p$ of the $\ell^p$ family and $E_{\min}$ all depend on $(\mu,\lambda)$ only through the unordered pair $\{|u|,|v|\}$, which both reflections preserve.
 
\begin{proposition}[The negation permutes the states]\label{prop:neg-states}
Let $E$ be a symmetric extremity index and let $\sigma$ be the classification of Definition~\ref{def:decision}. Define $\pi_{\neg}:S\to S$ by
\begin{align*}
    &t\leftrightarrow f,\qquad \top\mapsto\top,\qquad \bot\mapsto\bot,\\
    &qt\!\to\!\top\;\leftrightarrow\;qf\!\to\!\top,\qquad qt\!\to\!\bot\;\leftrightarrow\;qf\!\to\!\bot,\\
  &\top\!\to\!t\;\leftrightarrow\;\top\!\to\!f,\qquad
  \bot\!\to\!t\;\leftrightarrow\;\bot\!\to\!f.
\end{align*}
Then $\sigma(\neg a)=\pi_{\neg}\bigl(\sigma(a)\bigr)$ for every annotation $a$ with $\mu\neq\lambda$. Dually, with
\begin{align*}
    \pi_{-}:&\;t\mapsto t,\quad f\mapsto f,\quad \top\leftrightarrow\bot,\\
  &qt\!\to\!\top\leftrightarrow qt\!\to\!\bot,\quad
  qf\!\to\!\top\leftrightarrow qf\!\to\!\bot,\\
  &\top\!\to\!t\leftrightarrow\bot\!\to\!t,\qquad
  \top\!\to\!f\leftrightarrow\bot\!\to\!f,
\end{align*}
one has $\sigma(-a)=\pi_{-}\bigl(\sigma(a)\bigr)$ for every $a$ with $\mu+\lambda\neq1$. Both $\pi_{\neg}$ and $\pi_{-}$ are involutions of $S$, and they commute.
\end{proposition}
 
\begin{proof}
Write $u=\mu-\tfrac12$ and $v=\lambda-\tfrac12$, so that $\neg$ acts as $(u,v)\mapsto(v,u)$ and $-$ as $(u,v)\mapsto(-v,-u)$.
 
For the negation, assume $\mu\neq\lambda$, that is $u\neq v$. In Step~1 the clauses for $t$ and for $f$ are exchanged by $(u,v)\mapsto(v,u)$, and the clauses for $\top$ and for $\bot$ are each invariant; the two first clauses overlap only at $u=v=0$, which is excluded, so the dominant state is transformed as stated. Step~2 is unchanged because $E$ is symmetric. In Step~3, the quantity $\mu+\lambda$ is invariant under $\neg$, so a dominant $t$ becoming a dominant $f$ carries its tendency unchanged, giving $qt\!\to\!\top\mapsto qf\!\to\!\top$ and $qt\!\to\!\bot\mapsto qf\!\to\!\bot$; whereas $\mu-\lambda$ changes sign, so a dominant $\top$ or $\bot$ keeps its dominant state and reverses its tendency, the tie $\mu=\lambda$ being excluded by hypothesis.
 
For the conflation, assume $\mu+\lambda\neq1$. Step~1: $(u,v)\mapsto(-v,-u)$ fixes the clauses for $t$ and $f$ and exchanges those for $\top$ and $\bot$. Step~2 is again unchanged. In Step~3, $\mu-\lambda$ is invariant under $-$, so a dominant $\top$ or $\bot$ swaps while keeping its tendency; and $\mu+\lambda$ is replaced by $2-(\mu+\lambda)$, so the comparison with $1$ is reversed and a dominant $t$ or $f$ keeps its dominant state while reversing its tendency, the tie $\mu+\lambda=1$ being excluded.
 
That $\pi_{\neg}$ and $\pi_{-}$ are commuting involutions is read off the displayed assignments.
\end{proof}
 
\begin{corollary}\label{cor:orbits}
The group $\{\mathrm{id},\pi_{\neg},\pi_{-},\pi_{\neg}\pi_{-}\}$ acts on the twelve logical states with four orbits:
\begin{align*}
    &\{t,f\},\qquad \{\top,\bot\},\\
    &\{qt\!\to\!\top,\;qt\!\to\!\bot,\;qf\!\to\!\top,\;qf\!\to\!\bot\},\\
    &\{\top\!\to\!t,\;\top\!\to\!f,\;\bot\!\to\!t,\;\bot\!\to\!f\},
\end{align*}
of sizes $2$, $2$, $4$ and $4$.
\end{corollary}
 
\begin{proof}
Immediate from the assignments of Proposition~\ref{prop:neg-states}: the four extreme states split into the pair exchanged by $\pi_{\neg}$ and fixed by $\pi_{-}$ and the pair exchanged by $\pi_{-}$ and fixed by $\pi_{\neg}$, while each of the two families of non-extreme states is permuted simply transitively.
\end{proof}
 
\begin{remark}[The excluded lines]\label{rem:ties}
The two exceptions in Proposition~\ref{prop:neg-states} are exactly the fixed-point sets of the involutions: $\neg$ fixes the diagonal $\mu=\lambda$ and $-$ fixes the anti-diagonal $\mu+\lambda=1$. On those lines a point is its own image, so it cannot receive two different states, and the tie-breaking conventions of Definition~\ref{def:decision} --- which resolve $\mu\ge\lambda$ in favor of a tendency toward $t$, and $\mu+\lambda\ge1$ in favor of a tendency toward $\top$ --- decide the matter asymmetrically. The phenomenon is confined to two line segments of measure zero and does not affect any classification of a cohort; it could be removed at the cost of adding two boundary states, which we have preferred not to do.
\end{remark}
 
\begin{remark}[Reading]\label{rem:involutions-reading}
In the educational setting the two involutions have distinct meanings. Negation exchanges the roles of the two families of instruments, sending a student whose summative average is $\overline{x}$ and formative average $\overline{y}$ to a hypothetical student with the two figures exchanged; by Proposition~\ref{prop:reading} it reverses the sign of the degree of contradiction and leaves the degree of certainty unchanged, so it distinguishes the student who performs without engaging from the one who engages without performing. Conflation leaves the gap between the two averages untouched and reverses the overall level, so it maps a well-evidenced profile to a poorly evidenced one of the same character. Corollary~\ref{cor:orbits} then says that the twelve states are organized into four groups: the two verdicts, the two failures of evidence, and two quadruples in which each combination of a verdict with a failure occurs exactly once.
\end{remark}

\section{The Generalized Approach}\label{sec:generalized}

We propose a generalized approach that formally decouples the sources of evidence. In this approach, the degree of belief $\mu$ is strictly derived from cognitive proficiency, measured by summative assessments. Independently, the degree of disbelief $\lambda$ is derived from behavioral and attitudinal engagement, measured by formative evaluations. The two coordinates are therefore not two readings of one body of evidence, as in Section~\ref{sec:classical}, but two readings of two distinct bodies of evidence about the same proposition $P_i$.

Let $S=\{X_1,\dots,X_n\}$ be a set of summative variables (e.g., exams in Mathematics, Physics, Biology) and $F=\{Y_1,\dots,Y_m\}$ a set of formative variables (e.g., Teamwork, Prior Preparation, Class Engagement), the two sets being disjoint: no instrument contributes to both. For a given student $i$, let $x_{ij}$ and $y_{ik}$ denote the respective scores, strictly normalized to the unit interval $[0,1]$. To account for the relative importance of different disciplines or pedagogical criteria, we designate non-negative weights $p_1,\dots,p_n$ and $q_1,\dots,q_m$, with $\sum_j p_j>0$ and $\sum_k q_k>0$.

\begin{definition}[Generalized annotation]\label{def:generalized}
Write
$$\overline{x}_i=\frac{\sum_{j=1}^{n}p_j\,x_{ij}}{\sum_{j=1}^{n}p_j},
\qquad
\overline{y}_i=\frac{\sum_{k=1}^{m}q_k\,y_{ik}}{\sum_{k=1}^{m}q_k}$$
for the weighted summative and formative means of student $i$. The \emph{generalized annotation} of student $i$ is the pair
\begin{align}\label{non-classic}
  \mu_i:=\overline{x}_i,
  \qquad
  \lambda_i:=1-\overline{y}_i .
\end{align}
\end{definition}

The rationale behind these independent mappings is rooted in the pedagogical contract of active learning. The degree of belief $\mu_i$ operates as a direct weighted average of cognitive performance: high grades yield strong favorable evidence for $P_i$. The degree of disbelief $\lambda_i$, conversely, is the mathematical complement of formative performance. A student who demonstrates excellent teamwork and preparation achieves a high formative average and thereby generates a minimal degree of unfavorable evidence, $\lambda_i\to0$; a student who actively disrupts the collaborative environment or systematically fails to prepare generates strong unfavorable evidence, $\lambda_i\to1$. What makes the construction possible is a feature of the domain rather than of the model: summative and formative assessments are institutionally distinct instruments, applied at distinct moments and for distinct pedagogical reasons, so that neither average constrains the other.

\begin{proposition}\label{prop:generalized-onto}
The map $(\overline{x}_i,\overline{y}_i)\mapsto(\mu_i,\lambda_i)$ of Definition~\ref{def:generalized} is a homeomorphism of $[0,1]^2$ onto itself. In particular every annotation of the USCP is attainable; under the canonical choice $E=|u|+|v|$, $r=\tfrac12$, every one of the twelve logical states of Definition~\ref{def:decision} is attainable.
\end{proposition}

\begin{proof}
Each weighted mean is a convex combination of numbers in $[0,1]$ and therefore lies in $[0,1]$; conversely every value $c\in[0,1]$ is attained, for instance by a student all of whose scores in the family equal $c$. The map $(\overline{x},\overline{y})\mapsto(\overline{x},1-\overline{y})$ is a continuous involution of $[0,1]^2$, and therefore it is a homeomorphism. Under the canonical choice of $E$ and $r$ of Definition~\ref{def:extremity}, Proposition~\ref{prop:partition} ensures that each of the twelve regions has non-empty interior, each of them contains an attainable annotation.
\end{proof}

Proposition~\ref{prop:generalized-onto} is precisely the property that Corollary~\ref{cor:unreachable} denies to the classical pipeline: there the fused annotation was confined to half of the square, whereas here the whole square is in use. The next proposition explains what the two coordinates then measure.

\begin{proposition}[Reading of the coordinates]\label{prop:reading}
With the notation of Definition~\ref{def:generalized},
$$D_c=\mu_i-\lambda_i=\overline{x}_i+\overline{y}_i-1,
\qquad
D_{ct}=\mu_i+\lambda_i-1=\overline{x}_i-\overline{y}_i,$$
and the extremity index of Definition~\ref{def:decision} is
$$E_i=\Bigl|\overline{x}_i-\tfrac12\Bigr|+\Bigl|\overline{y}_i-\tfrac12\Bigr|.$$
Moreover the quadrant of the annotation, and hence its dominant logical state, is determined by the position of the two means relative to the midpoint of the scale: the state is $t$ when $\overline{x}_i\ge\tfrac12$ and $\overline{y}_i\ge\tfrac12$; $f$ when $\overline{x}_i\le\tfrac12$ and $\overline{y}_i\le\tfrac12$; $\top$ when $\overline{x}_i>\tfrac12>\overline{y}_i$; and $\bot$ when $\overline{y}_i>\tfrac12>\overline{x}_i$, with ties broken by the order of the clauses in Definition~\ref{def:decision}.
\end{proposition}

\begin{proof}
Substituting $\mu_i=\overline{x}_i$ and $\lambda_i=1-\overline{y}_i$ gives $D_c=\overline{x}_i-(1-\overline{y}_i)=\overline{x}_i+\overline{y}_i-1$ and $D_{ct}=\overline{x}_i+(1-\overline{y}_i)-1=\overline{x}_i-\overline{y}_i$. For the extremity index, $|\lambda_i-\tfrac12|=|1-\overline{y}_i-\tfrac12|=|\tfrac12-\overline{y}_i|=|\overline{y}_i-\tfrac12|$, and the stated formula follows. Finally, the quadrant conditions of Definition~\ref{def:decision} are $\mu_i\gtrless\tfrac12$ and $\lambda_i\gtrless\tfrac12$; the first is $\overline{x}_i\gtrless\tfrac12$, and $\lambda_i\ge\tfrac12$ is equivalent to $\overline{y}_i\le\tfrac12$, which yields the four cases as stated.
\end{proof}

Proposition~\ref{prop:reading} is what makes the diagram usable by an instructor who has never seen a paraconsistent lattice. The horizontal position of a student records conceptual mastery and the vertical position records disengagement. The degree of contradiction is exactly the gap $\overline{x}_i-\overline{y}_i$ between what the student demonstrates in examinations and what the student invests in the course --- and it is worth contrasting this with the classical pipeline, where the same quantity measured the disagreement among sources. A student is classified as extreme exactly when both dimensions are far from the midpoint of the scale. The four quadrants then carry the four readings recorded in Table~\ref{tab:states}: strong on both dimensions, weak on both, strong in examinations but disengaged, or engaged but not yet succeeding in examinations.

\begin{remark}[The classical pipeline is a special case]\label{rem:special-case}
Definition~\ref{def:generalized} generalizes the usual construction rather than replacing it. If the formative family is taken to be the summative one, that is $m=n$, $Y_k=X_k$ and $q_k=p_k$ for every $k$, then $\overline{y}_i=\overline{x}_i=\mu_i$ and \eqref{non-classic} collapses to $\lambda_i=1-\mu_i$, the classical annotation, whose degree of contradiction vanishes identically. The states excluded by Corollary~\ref{cor:unreachable} are therefore excluded precisely because the two coordinates are being read off the same evidence, and they become available as soon as the two families are genuinely distinct.
\end{remark}

\begin{remark}[On the weights]\label{rem:weights}
Definition~\ref{def:generalized} leaves the weights free, so that a curriculum may encode in them the relative importance it assigns to each instrument. Two choices deserve mention. The uniform one, $p_j=1$ and $q_k=1$, reduces both coordinates to unweighted means. The one adopted below takes each weight to be the standard deviation of the corresponding variable within the cohort, so that an instrument counts in proportion to how much it separates students: a discipline in which everyone scores alike carries little information about any individual, whereas one that spreads the cohort out carries a great deal. This is a familiar idea in psychometrics, and it makes the construction applicable to real assessment records, where instruments differ widely in discriminative power.
\end{remark}

\section{Simulation Design and Data Generating Process}

To validate the proposed methodology, we conducted a simulation study to evaluate the PAL2v classification model under controlled parametric cohorts. We simulated raw assessment scores on a standard educational scale of $[0, 10]$ for $N=400$ students per cohort. Throughout the simulation the extremity boundary is the canonical one of Definition~\ref{def:extremity}, $E=|u|+|v|$ with $r=\tfrac12$. The freedom left by that definition is exercised here only to the extent of fixing this choice explicitly; the behavior of the classification under other boundaries is a question of calibration against assessment practice, and we leave it aside.

To mirror the assessment structure of rigorous exact sciences courses (such as Calculus and Algebra), the simulated database features variables bifurcated into two independent dimensions. The cognitive dimension is captured by $n=4$ summative assessments, corresponding to the disciplines of Mathematics, Physics, Biology and Chemistry. The behavioral dimension is captured by $m=3$ formative evaluations inherent to active learning methodologies such as Team-Based Learning, namely Teamwork, Prior Preparation, Class Engagement. The abbreviations are those used in Table~\ref{tab:estatistica_descritiva}.

For each simulated student $i$, raw assessment data were generated using independent stochastic processes. To control the distribution shapes while respecting the bounded nature of educational grading, we used scaled Beta distributions:

\begin{enumerate}
    \item \textbf{Summative Assessments (Raw Grades):} A tuple of $n$ summative grades $x_i = (x_{i1}, \dots, x_{in}) \in [0,10]^n$ was generated, where each score is calculated as $10X$, with the random variable $X \sim \text{Beta}(\alpha_{x}, \beta_{x})$.
    \item \textbf{Formative Assessments (Raw Behavioral Scores):} A tuple of $m$ formative evaluations $y_i = (y_{i1}, \dots, y_{im}) \in [0,10]^m$ was generated, where each score is calculated as $10Y$, with the random variable $Y \sim \text{Beta}(\alpha_{y}, \beta_{y})$.
\end{enumerate}

The shape parameters were fixed per cohort as follows. In Cohort A a single latent variable $L\sim\mathrm{Beta}(5,2)$ drives both families, each score being $10L$ perturbed by independent Gaussian noise of standard deviation $0.5$ and truncated to $[0,10]$. In Cohort B the summative scores are drawn from $\mathrm{Beta}(7,2)$ and the formative ones from $\mathrm{Beta}(8,2)$, independently. In Cohort C each summative variable is drawn from an equal mixture of $\mathrm{Beta}(2,6)$ and $\mathrm{Beta}(8,3)$, and each formative one from $\mathrm{Beta}(3,4)$. All scores are rounded to one decimal place. The pseudo-random seed was fixed at $42$ and $N=400$ students were generated per cohort.

The choice of the Beta distribution is mathematically motivated. Its natural support bounded within the interval $[0,1]$ maps perfectly to closed grading systems, avoiding the theoretical inconsistencies of using distributions with infinite tails (such as the Gaussian), which would require ad hoc truncation. Gaussian noise was restricted exclusively to the stochastic perturbation in the control cohort (Cohort A). Furthermore, the parametric flexibility provided by the shape parameters $\alpha$ and $\beta$ allows for the modeling of diverse real-world classroom profiles, ranging from highly skewed distributions (representing high-achieving or struggling cohorts) to bimodal shapes. The single exception is Cohort A, where the dependence between the two families is induced by adding Gaussian perturbations to a common latent Beta variable; the resulting scores are truncated to $[0,10]$, which affects fewer than one per cent of them.

The variables were mapped into the unit interval required by the PAL2v lattice by linear scaling, $x_{ij}/10$ and $y_{ik}/10$. Following Remark~\ref{rem:weights}, each weight was taken to be the standard deviation of the corresponding variable within the cohort, $p_j=s(X_j)$ and $q_k=s(Y_k)$, so that the paraconsistent annotation of student $i$ is
$$ \mu_i = \frac{\sum_{j=1}^{n} s(X_j)\,x_{ij}/10}{\sum_{j=1}^{n} s(X_j)}
\quad \mbox{and} \quad
\lambda_i = 1 - \frac{\sum_{k=1}^{m} s(Y_k)\,y_{ik}/10}{\sum_{k=1}^{m} s(Y_k)} .$$
The resulting weights are reported in Table~\ref{tab:weights}. Because every summative variable is drawn from the same distribution, and likewise every formative one, the weights are close to uniform in all three cohorts and the classification changes in fewer than five per cent of the students when they are replaced by uniform ones. The simulation therefore isolates the effect of the decoupling, which is what the article is about; the weights matter for real assessment records, where instruments differ in dispersion, and not for synthetic ones built to be homogeneous.

\begin{table}[htpb]
\centering
\caption{Weights used in the generalized construction, taken as the within-cohort standard deviation of each variable and shown here normalized to sum to one within each family}
\label{tab:weights}
\small
\begin{tabular}{@{}lrrrrrrr@{}}
\toprule
Cohort & Mathematics & Physics & Chemistry & Biology & Teamwork & Prior Preparation & Class Engagement \\
\midrule
A & 0.253 & 0.251 & 0.252 & 0.244 & 0.334 & 0.334 & 0.332 \\
B & 0.247 & 0.271 & 0.242 & 0.239 & 0.340 & 0.322 & 0.338 \\
C & 0.252 & 0.247 & 0.251 & 0.251 & 0.334 & 0.354 & 0.312 \\
\bottomrule
\end{tabular}
\end{table}

\begin{remark}[Graphical scaling]\label{rem:scaling}
Annotations are defined on the unit square, and every definition, proposition and table in this article is stated there. In the scatter plots of Section~\ref{sec:results}, however, both coordinates are multiplied by $10$, so that the axes carry the $[0,10]$ scale on which the curriculum records its assessments. The correspondence is immediate: the two medians of Step~1 read $\mu=5$ and $\lambda=5$, the two diagonals of Step~3 read $\mu=\lambda$ and $\mu+\lambda=10$, and the canonical extremity boundary of Definition~\ref{def:extremity} reads $|\mu-5|+|\lambda-5|\ge5$. No result depends on the scale; the choice is made so that a reader who works with the assessment records can locate a student on the diagram without converting anything.
\end{remark}

By establishing these decoupled generative processes for $\mu$ and $\lambda$, we avoid the deterministic constraint $\lambda = 1 - \mu$. We examine three distinct cohort profiles under the generalized procedure:

\begin{itemize}
    \item \textbf{Cohort A (Classical Bias):} A control cohort designed by inducing high covariance between summative and formative scores through a shared latent random variable. This stochastic dependence makes the two weighted means nearly coincide, so that by Proposition~\ref{prop:reading} the degree of contradiction is close to zero. The cohort is a control: it shows what the generalized construction returns when the two families of instruments happen to carry the same information, and should not be confused with the classical pipeline, whose confinement is a property of the aggregation and holds for every input.
    
    \item \textbf{Cohort B (High Consistency Cohort):} Parameterized to reflect an ideal class with strong academic and behavioral scores. Strongly left-skewed, independent Beta distributions ($\alpha > \beta$) concentrate the simulated coordinate pairs securely within the True ($t$) state of the unit square.
    
    \item \textbf{Cohort C (Educational Paradoxes):} Parameterized to force the emergence of highly inconsistent profiles, exploring the non-trivial regions. Each summative variable is drawn from a bimodal mixture, so that any single discipline separates students with severe conceptual difficulties from those with excellent mastery. The four disciplines are drawn independently of one another, which is the realistic assumption for a cohort in which difficulty is subject-specific rather than global; the composite degree of belief is therefore dispersed rather than bimodal.
\end{itemize}

For the comparison reported in Table~\ref{tab:comparison}, the classical pipeline was reconstructed on the same data by treating each of the seven instruments as a source with $\lambda_k=1-\mu_k$ and fusing the resulting annotations by the componentwise minimum, which is the procedure described in the applied literature and analyzed in Section~\ref{sec:classical}. This is a modeling choice on our part, since the literature does not fix what counts as a source in an educational setting; we adopt the most direct reading, in which every recorded instrument contributes one annotation.

A scatter plot containing the raw simulated data for each cohort is shown in Fig. \ref{fig:dataset}, and the descriptive statistics characterizing the generated variables are detailed in Table \ref{tab:estatistica_descritiva}.

\begin{figure}[htpb]
    \centering
    \includegraphics[width=0.5\linewidth]{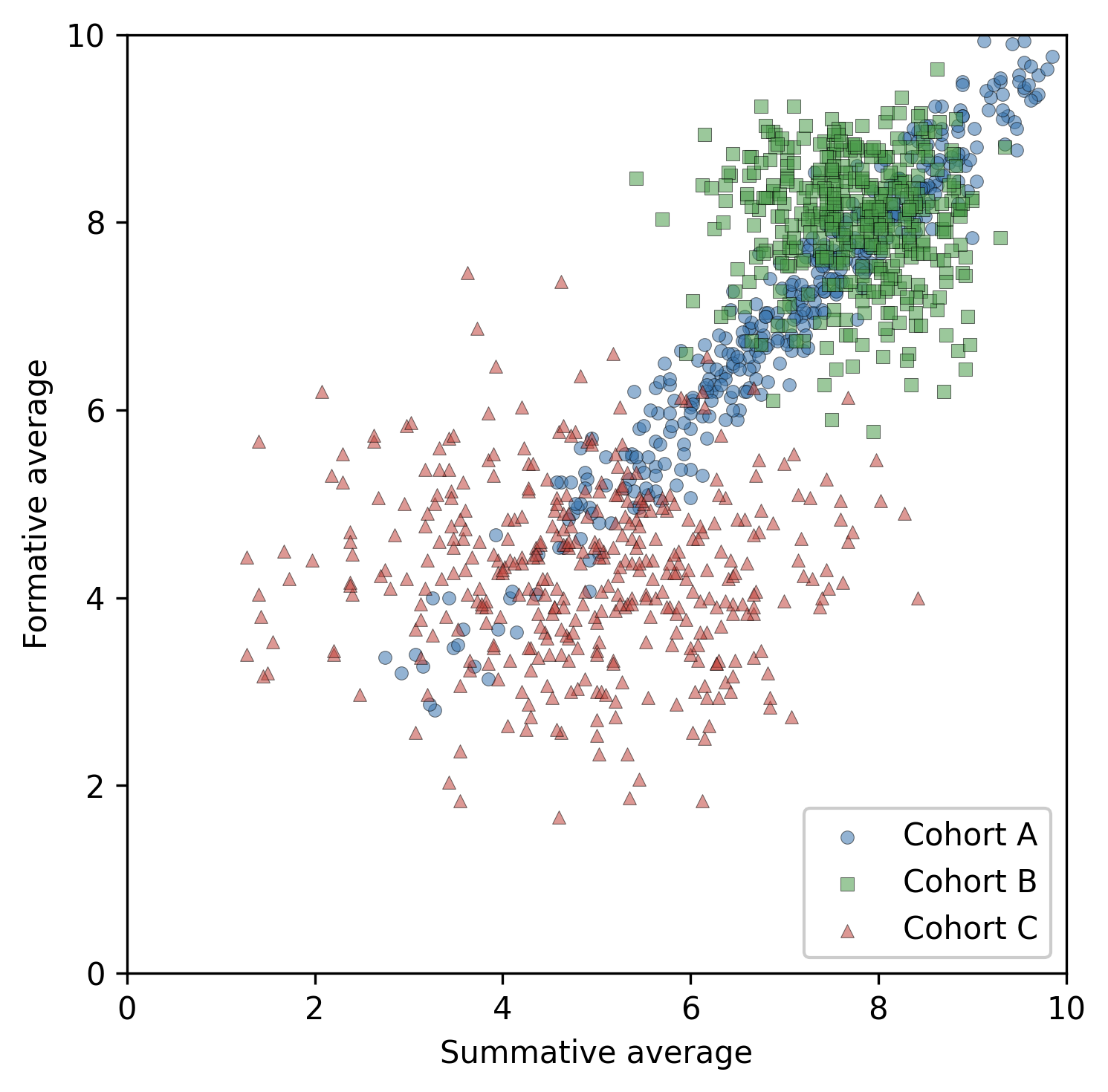}
    \caption{Scatter plot of summative and formative assessments per cohort.}
    \label{fig:dataset}
\end{figure}

\begin{table}[htpb]
\centering
\caption{Descriptive statistics of the simulated raw scores, on the $[0,10]$ scale, by cohort and overall}
\label{tab:estatistica_descritiva}
\small
\begin{tabular}{@{}llrrrrrr@{}}
\toprule
Cohort & Variable & $N$ & Mean & SD & Median & Min & Max \\
\midrule
A & Mathematics & 400 & 7.11 & 1.56 & 7.20 & 2.10 & 10.00 \\
 & Physics & 400 & 7.06 & 1.55 & 7.10 & 2.40 & 10.00 \\
 & Chemistry & 400 & 7.10 & 1.56 & 7.20 & 2.30 & 10.00 \\
 & Biology & 400 & 7.07 & 1.51 & 7.20 & 2.60 & 10.00 \\
 & Teamwork & 400 & 7.11 & 1.55 & 7.30 & 2.70 & 10.00 \\
 & Prior Preparation & 400 & 7.08 & 1.56 & 7.30 & 2.80 & 10.00 \\
 & Class Engagement & 400 & 7.10 & 1.54 & 7.15 & 2.60 & 10.00 \\
\midrule
B & Mathematics & 400 & 7.77 & 1.31 & 7.95 & 3.60 & 9.90 \\
 & Physics & 400 & 7.62 & 1.43 & 7.80 & 3.20 & 10.00 \\
 & Chemistry & 400 & 7.80 & 1.28 & 7.90 & 2.50 & 10.00 \\
 & Biology & 400 & 7.76 & 1.27 & 8.00 & 2.80 & 9.90 \\
 & Teamwork & 400 & 7.95 & 1.20 & 8.20 & 3.70 & 9.90 \\
 & Prior Preparation & 400 & 8.08 & 1.13 & 8.20 & 3.20 & 10.00 \\
 & Class Engagement & 400 & 7.98 & 1.19 & 8.25 & 4.10 & 9.80 \\
\midrule
C & Mathematics & 400 & 4.79 & 2.79 & 4.75 & 0.10 & 9.90 \\
 & Physics & 400 & 4.93 & 2.73 & 5.10 & 0.10 & 9.50 \\
 & Chemistry & 400 & 4.93 & 2.78 & 5.10 & 0.20 & 9.70 \\
 & Biology & 400 & 4.92 & 2.78 & 5.00 & 0.10 & 9.50 \\
 & Teamwork & 400 & 4.33 & 1.73 & 4.40 & 0.40 & 8.50 \\
 & Prior Preparation & 400 & 4.45 & 1.83 & 4.50 & 0.60 & 9.00 \\
 & Class Engagement & 400 & 4.10 & 1.61 & 4.15 & 0.70 & 8.60 \\
\midrule
All & Mathematics & 1200 & 6.56 & 2.37 & 7.10 & 0.10 & 10.00 \\
 & Physics & 1200 & 6.54 & 2.31 & 7.10 & 0.10 & 10.00 \\
 & Chemistry & 1200 & 6.61 & 2.33 & 7.20 & 0.20 & 10.00 \\
 & Biology & 1200 & 6.59 & 2.31 & 7.20 & 0.10 & 10.00 \\
 & Teamwork & 1200 & 6.46 & 2.16 & 6.85 & 0.40 & 10.00 \\
 & Prior Preparation & 1200 & 6.54 & 2.17 & 7.00 & 0.60 & 10.00 \\
 & Class Engagement & 1200 & 6.39 & 2.21 & 6.80 & 0.70 & 10.00 \\
\bottomrule
\end{tabular}
\end{table}

\section{Results}\label{sec:results}

The simulation results confirm that the generalized PAL2v construction correctly classifies non-trivial student profiles. Table \ref{tab:estatistica_descritiva} provides the descriptive statistics for all simulated variables, confirming that the data generation process produced the intended distributional shapes for each educational cohort. The spatial classification of these cohorts onto the USCP is represented in Figure~\ref{fig:three_figures}, on the $[0,10]$ axes of Remark~\ref{rem:scaling}.

\begin{figure}[htpb]
    \centering

    \begin{subfigure}{0.32\textwidth}
        \centering
        \includegraphics[width=\linewidth]{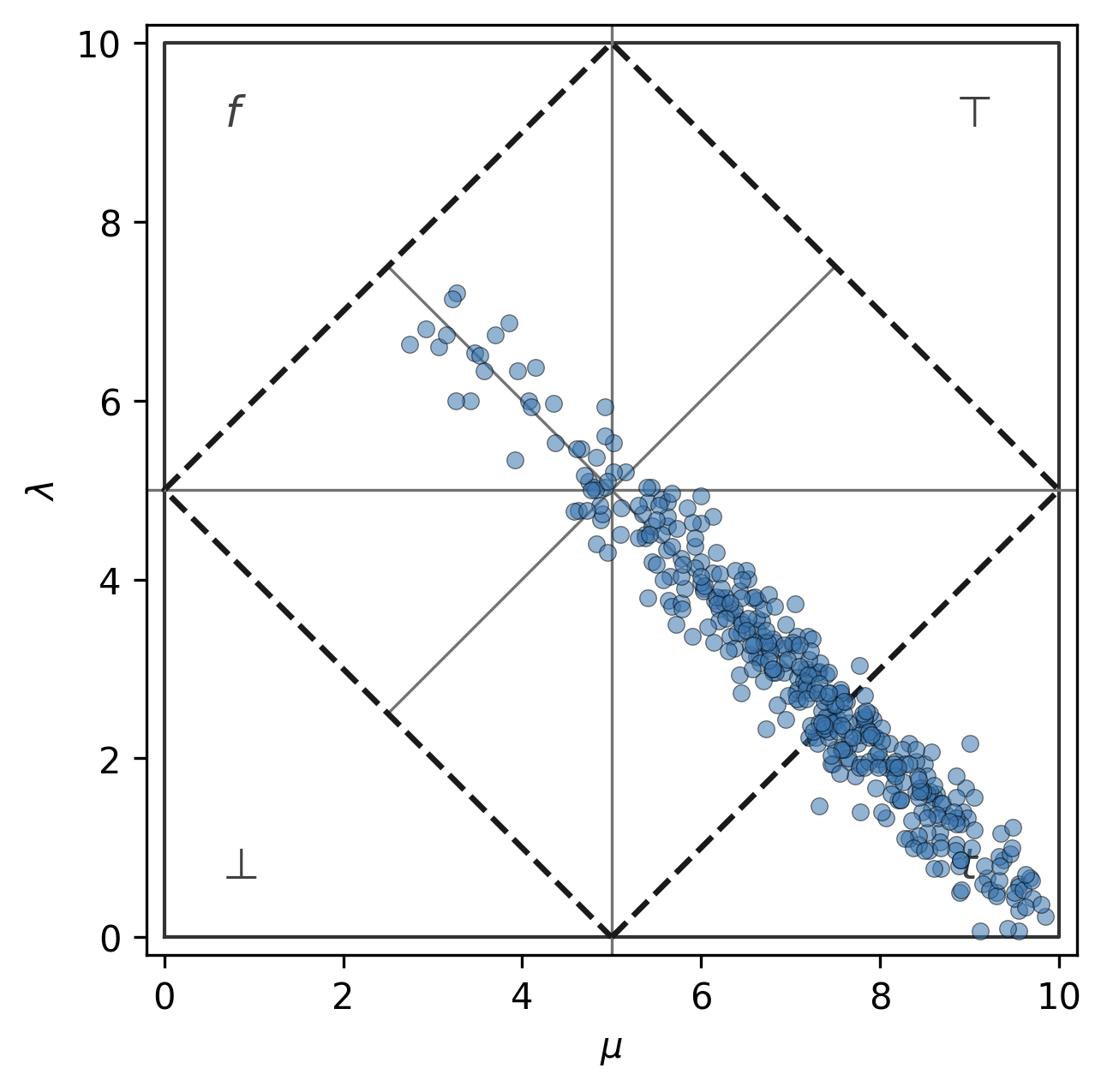}
        \caption{Cohort A}
    \end{subfigure}
    \hfill
    \begin{subfigure}{0.32\textwidth}
        \centering
        \includegraphics[width=\linewidth]{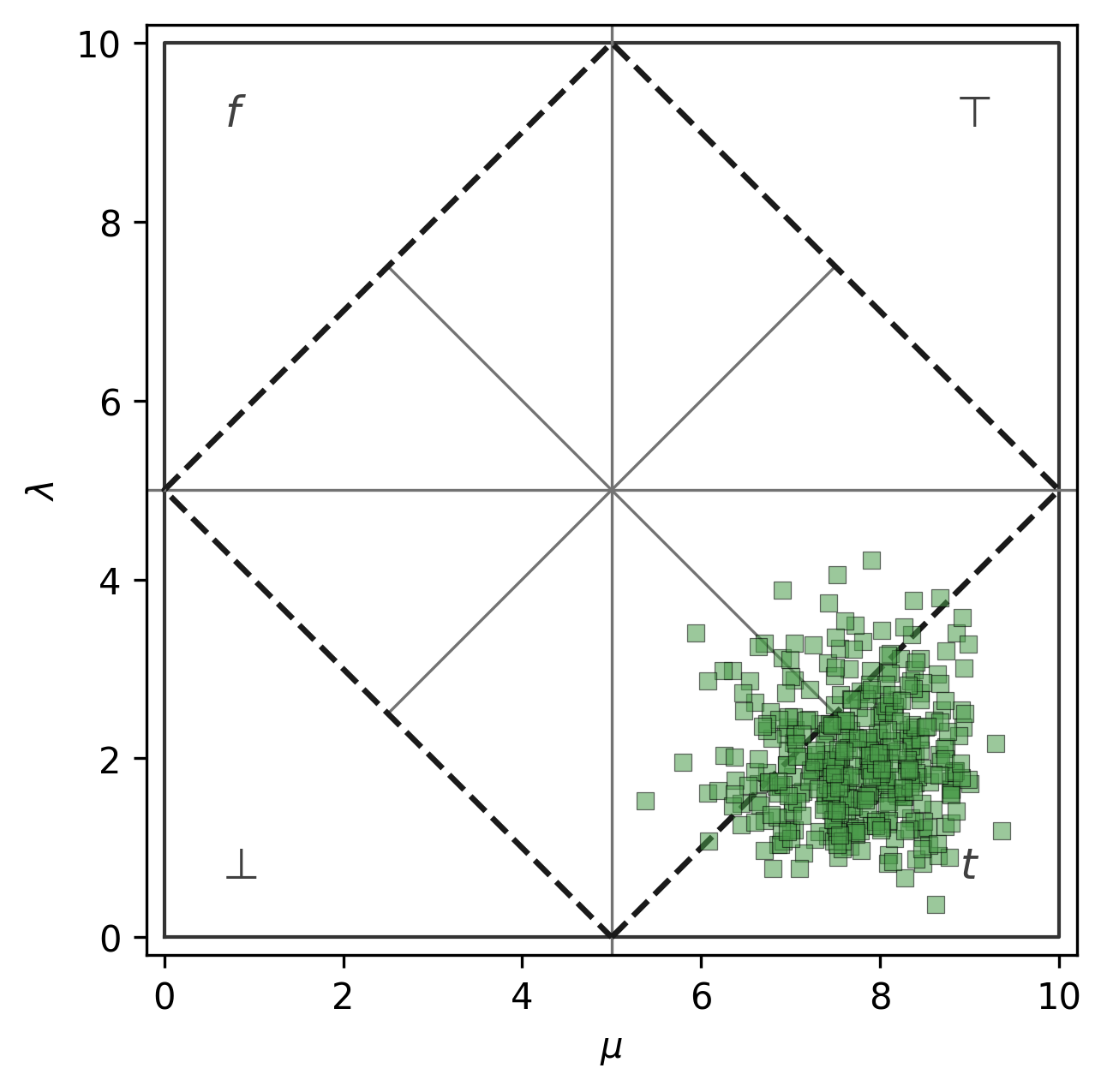}
        \caption{Cohort B}
    \end{subfigure}
    \hfill
    \begin{subfigure}{0.32\textwidth}
        \centering
        \includegraphics[width=\linewidth]{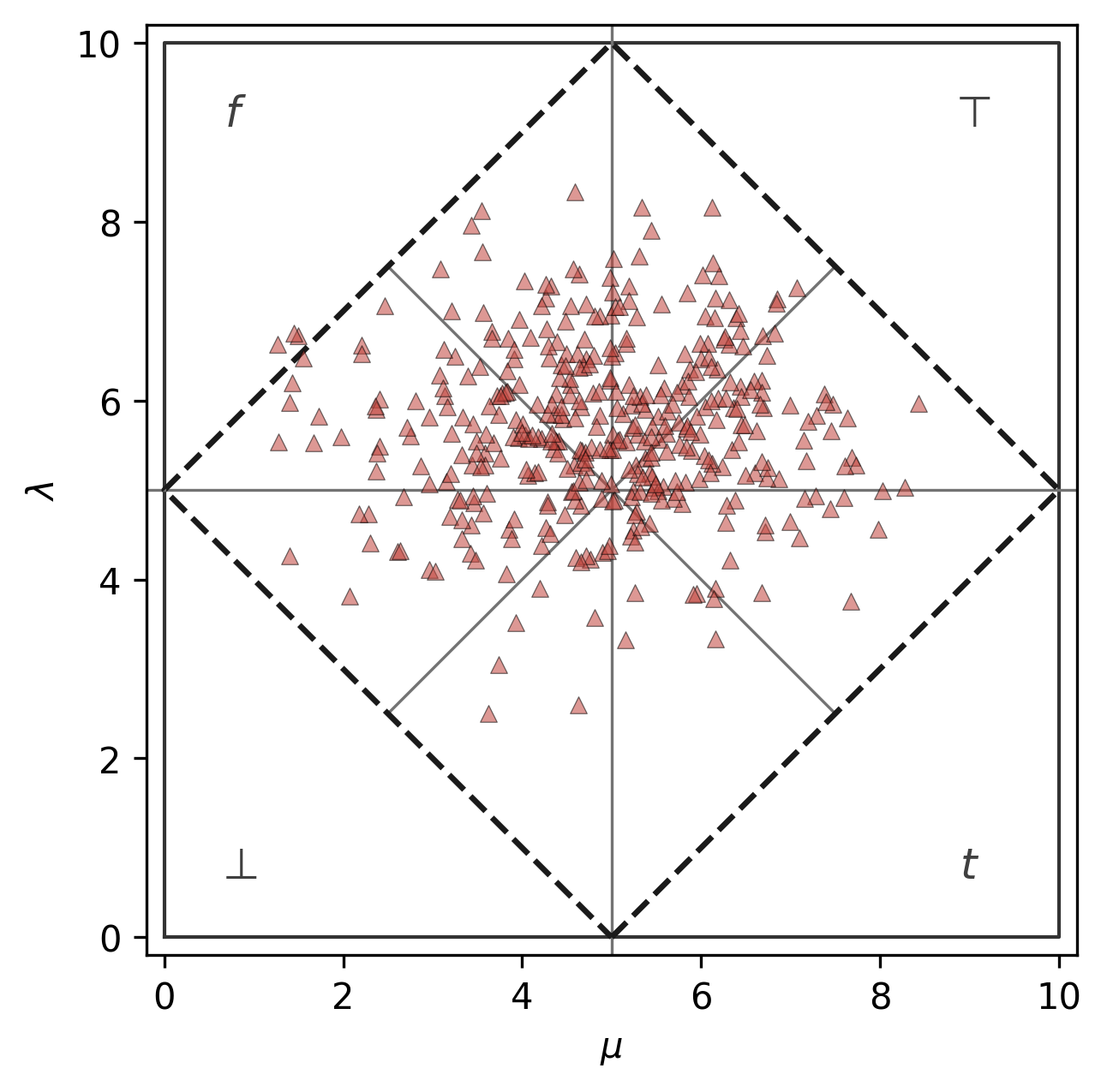}
        \caption{Cohort C}
    \end{subfigure}
    \caption{Annotations $(\mu,\lambda)$ of the three simulated cohorts, obtained throughout by the generalized construction of Definition~\ref{def:generalized} and plotted on the $[0,10]$ axes of Remark~\ref{rem:scaling}. (a) Cohort A, in which the two families of instruments are driven by a common latent variable, so that the annotations concentrate near the anti-diagonal. (b) Cohort B, a high-consistency cohort. (c) Cohort C, with summative and formative assessments generated independently. The dashed contour is the canonical extremity boundary $|\mu-5|+|\lambda-5|=5$.}
    \label{fig:three_figures}
\end{figure}

Figure~\ref{fig:three_figures}(a) illustrates the outcome of the generalized construction applied to Cohort A. Because both families of instruments are driven by a common latent variable, the two weighted means nearly coincide, and by Proposition~\ref{prop:reading} the degree of contradiction $D_{ct}=\overline{x}_i-\overline{y}_i$ is close to zero: the annotations concentrate along the anti-diagonal. Two points about this panel deserve attention. The construction used here is the generalized one throughout; Cohort A merely happens to have strongly correlated evidence sources, and its annotations therefore lie near the line to which the classical pipeline is confined. The confinement established in Corollary~\ref{cor:unreachable} is of a different nature: it holds for every input, by Proposition~\ref{prop:duality}, and is a property of the aggregation rather than of the data.

Figure \ref{fig:three_figures}(b) displays the coordinate mapping for Cohort B, parameterized to represent a high-consistency class. As expected from the independent, strongly left-skewed distributions, the annotations concentrate in the lower right quadrant, where summative performance is high and unfavorable evidence is low, so that the dominant state is $t$ throughout.

Figure \ref{fig:three_figures}(c) illustrates the paradoxical profiles of Cohort C. By using a bimodal distribution for summative assessments and an independent distribution for formative evaluations, the generalized approach scatters the annotations across the entire two-dimensional lattice space. Unlike Cohort A, the annotations now occupy all four quadrants: $10.2\%$ fall in the quadrant of truth, $39.4\%$ in that of falsity, $38.0\%$ in the inconsistent one and $12.4\%$ in the paracomplete one. Students in the inconsistent quadrant exhibit a high degree of belief, from excellent summative grades, together with a high degree of disbelief, from poor formative engagement; those in the paracomplete quadrant exhibit a low degree of belief together with a low degree of disbelief, that is, weak examination performance with adequate engagement. Almost none of them, however, is extreme: $99.2\%$ of the cohort lies strictly inside the canonical extremity boundary, so the states assigned are qualified rather than outright. The cohort displays the profiles the construction is meant to expose, without the evidence being decisive enough for an unqualified verdict.

Table~\ref{tab:states-by-cohort} records the distribution of the twelve logical states across the three cohorts. Three features are worth noting. Cohort B occupies three states only --- $t$, $qt\!\to\!\top$ and $qt\!\to\!\bot$ --- which is the expected behavior of a control: where the evidence is uniformly strong and consistent, the construction returns a verdict and invents no contradiction. Cohort A is concentrated in the same corner but spreads over nine states, a point we return to below. Cohort C, by contrast, leaves the corner altogether: its four most populated states are $\top\!\to\!f$ at $20.2\%$, $qf\!\to\!\bot$ at $19.8\%$, $qf\!\to\!\top$ at $18.8\%$ and $\top\!\to\!t$ at $17.8\%$,  three of which are among the states the classical pipeline cannot express. No student in any cohort is left unclassified, as Proposition~\ref{prop:partition} guarantees.

\begin{table}[htpb]
\centering
\caption{Distribution of the twelve logical states by cohort, in per cent, under the generalized construction with standard-deviation weights and the canonical extremity boundary}
\label{tab:states-by-cohort}
\small
\begin{tabular}{@{}lrrrrrr@{}}
\toprule
Cohort & $t$ & $f$ & $\top$ & $\bot$ & $qt\!\to\!\top$ & $qt\!\to\!\bot$ \\
\midrule
A & 43.2 & 0.0 & 0.0 & 0.0 & 24.2 & 21.2 \\
B & 81.0 & 0.0 & 0.0 & 0.0 & 7.0 & 12.0 \\
C & 0.0 & 0.8 & 0.0 & 0.0 & 7.0 & 3.2 \\
\bottomrule
\end{tabular}
\\[1ex]
\begin{tabular}{@{}lrrrrrr@{}}
\toprule
Cohort & $qf\!\to\!\top$ & $qf\!\to\!\bot$ & $\top\!\to\!t$ & $\top\!\to\!f$ & $\bot\!\to\!t$ & $\bot\!\to\!f$ \\
\midrule
A & 4.2 & 3.5 & 0.5 & 0.8 & 1.2 & 1.0 \\
B & 0.0 & 0.0 & 0.0 & 0.0 & 0.0 & 0.0 \\
C & 18.8 & 19.8 & 17.8 & 20.2 & 3.2 & 9.2 \\
\bottomrule
\end{tabular}
\\[1ex]
\end{table}

Cohort A deserves a second look. It was built so that a single latent variable drives both families of instruments, and the annotations do concentrate along the anti-diagonal, as Figure~\ref{fig:three_figures}(a) shows. Yet $29.7\%$ of the cohort still falls in the five states that require $\mu+\lambda\ge1$, and only $88.6\%$ of it lands in the quadrant of truth. The reason is that the two families are measured with independent error: even perfectly correlated underlying abilities yield summative and formative averages that differ by a little, and the degree of contradiction records that difference. This is worth stating because it is a property of measurement rather than of the simulation. Whenever two instruments estimate the same quantity imperfectly, some residual contradiction is present, and a construction that sets $\lambda=1-\mu$ discards it not because the evidence agrees, but because the arithmetic has been arranged so that it cannot disagree.

Table~\ref{tab:comparison} runs the two constructions on the same data. The second column reports the proportion of annotations lying in the five states that require $\mu+\lambda\ge1$; the third reports the same quantity under the classical pipeline. It is zero in all three cohorts, without a single exception among the $1200$ simulated students, which is Corollary~\ref{cor:unreachable} exhibited on data rather than proved. The generalized construction, on the same scores, places $29.7\%$, $7.0\%$ and $63.8\%$ of each cohort in those states. The fourth column shows how far the two disagree: between a third and a half of the students change logical state in Cohorts A and B, and almost all of them do in Cohort C. The disagreement is not a matter of degree at the margins; it is the difference between a classification that can express conflicting evidence and one that cannot.

\begin{table}[htbp]
\centering
\caption{The generalized construction against the classical pipeline on the same simulated data. The second and third columns report the proportion of annotations falling in the five states that require $\mu+\lambda\ge1$}
\label{tab:comparison}
\small
\begin{tabular}{@{}lrrrr@{}}
\toprule
Cohort & Above anti-diagonal & Above anti-diagonal & Reclassified & Extreme \\
 & (generalized, \%) & (classical, \%) & (\%) & (generalized, \%) \\
\midrule
A & 29.7 & 0.0 & 36.8 & 43.2 \\
B & 7.0 & 0.0 & 36.2 & 81.0 \\
C & 63.8 & 0.0 & 97.8 & 0.8 \\
\bottomrule
\end{tabular}
\end{table}

\section{Discussion}

Decoupling the two coordinates returns the paraconsistent lattice to use. Where the classical pipeline confines every annotation to half the square, leaving five of the twelve states unreachable in the non-degenerate case, the construction proposed here reaches all of them, and does so because summative and formative assessments are separate instruments rather than two readings of one. The gap between what a student demonstrates in examinations and what that student invests in the course becomes a coordinate of the representation instead of disappearing into an average.

The spatial classification within the USCP informs targeted pedagogical intervention by identifying profiles that a single average obscures. In Cohort C, the students who fall in the paracomplete quadrant --- $12.4\%$ of the cohort, distributed over the states $\bot\!\to\!t$ and $\bot\!\to\!f$ --- are the ``hardworking but struggling'' case: they meet the behavioral contract, preparing and collaborating, but have not yet grasped the core concepts. The appropriate response is academic support, such as foundational mentoring or tutoring, rather than any disciplinary measure.

Conversely, the $38.0\%$ falling in the inconsistent quadrant, chiefly in the states $\top\!\to\!f$ and $\top\!\to\!t$, are the ``brilliant but uncooperative'' case: high summative marks together with weak preparation or disruptive collaborative behavior, and therefore a high degree of disbelief. Under a standard grade point average the examination marks would mask the behavioral deficit entirely. Here the contradiction is visible in the coordinate itself and acts as an alert for an attitudinal intervention, which is what preserves the pedagogical integrity of an active learning methodology such as Team-Based Learning. That most of these students receive a qualified rather than an outright state is itself informative: it says the evidence points to the profile without being decisive, which is the right posture for a system meant to prompt a conversation rather than to issue a verdict.

Presenting the states on the square itself, rather than through the scalar indices $D_c$ and $D_{ct}$, keeps the representation legible to the people who have to act on it. An instructor reads a student's position without converting coordinates: the quadrant gives the dominant state, the distance from the center says whether the evidence is decisive, and the two dimensions stay separately visible instead of being summed into a single mark.

We note two limitations. The structural claims --- that the classical pipeline confines annotations to half the lattice, and that the twelve regions admit an unambiguous specification --- are established by proof and do not depend on data. Everything else rests on simulation. Synthetic scores can show that the states released by the decoupling are populated under plausible cohort profiles; they cannot show that those states are populated in any real classroom, nor how often. That requires assessment records from actual cohorts, and it is the next step. Whether the resulting classifications track anything of longer-term interest, such as retention or professional development, is a further question that only a longitudinal study could answer.

\section{Conclusion}\label{sec13}

This article identifies a structural limitation in the applied literature on Paraconsistent Annotated Logic with two values. When the degree of disbelief is taken as the complement of the degree of belief, the sign of the degree of contradiction is settled by the choice of aggregation rather than by the evidence, and the fused annotation cannot leave one half of the lattice. Five of the twelve logical states, among them contradiction itself, are then unreachable in the non-degenerate case, by construction.

We proposed a decoupled construction in which the degree of belief comes from summative assessments and the degree of disbelief from formative ones, and showed that every annotation of the square, and therefore every logical state, is then attainable. We also gave an unambiguous specification of the twelve regions, which the inequality lists in current use do not provide. The simulations illustrate the profiles that the construction makes visible and that a single average conceals: the student who performs well in examinations while withdrawing from the collaborative work of the course, and the student who does the work without yet succeeding in the examinations.

The next step is empirical. Applying the construction to assessment records from real cohorts would show which states are actually occupied, and how the occupancy moves over a term; a longitudinal design would then allow the classifications to be compared against outcomes such as retention. Neither question can be settled by
simulation, and both are open.

\vspace{0.5cm}

{\bf{Data and code availability}}

The data analyzed in this article are synthetic. They are produced by the script \texttt{generate\_data.py}, which fixes the pseudo-random seed at $42$ and generates $N=400$ students per cohort; the annotations, tables and figures are produced by \texttt{analysis.py}. Both scripts are available as supplementary material, and running them in sequence reproduces every number and every figure reported here. No human subjects were involved and no personal data were processed, so no ethics approval was required.

\vspace{0.3cm}

{\bf{Declaration of assistance}}

An AI language model was used to assist in drafting the structure of the data-generating script. All distributions, parameters and mathematical results were defined, verified and are the responsibility of the authors. No text or data in this article was generated by such a model without author verification.



\bibliographystyle{plain}
\bibliography{cluster-analysis}
\end{document}